\documentclass[10pt]{iopart}

\usepackage{iopams} 

\usepackage{graphicx}
\usepackage{amsthm}
\newcommand{\norm}[1]{\left\| #1 \right\|}
\newcommand{\ainnerproduct}[2]{\langle #1, #2 \rangle}
\usepackage{xcolor}
\newtheorem{definition}{Definition}
\newtheorem{theorem}{Theorem}
\newtheorem{proposition}{Proposition}
\newtheorem{lemma}{Lemma}
\newcommand{\bef}{\begin{figure}}
\newcommand{\eef}{\end{figure}}
\newcommand{\bseq}{\begin{subequations}}
\newcommand{\eseq}{\end{subequations}}
\usepackage{subfig}

\usepackage{cleveref}
\usepackage{etoolbox}

\patchcmd{\numparts}{\addtocounter{equation}{1}}{\refstepcounter{equation}}{}{}

\begin{document}

\title[Quantifying macrostructures in viscoelastic sub-diffusive flows]{Quantifying macrostructures in viscoelastic sub-diffusive flows}

\author{T Chauhan$^{1*}$, K Kalyanaraman$^1$ and S Sircar$^{1}$}

\address{$^1$ Department of Mathematics, IIIT Delhi, India 110020}
\address{$^*$ Corresponding Author}
\eads{\mailto{tanishai@iiitd.ac.in, \mailto{kaushik@iiitd.ac.in}, \mailto{sarthok@iiitd.ac.in}}}
\vspace{10pt}

\begin{abstract}
We present a theory to quantify the formation of spatiotemporal macrostructures (or the non-homogeneous regions of high viscosity at moderate to high fluid inertia) for viscoelastic sub-diffusive flows, by introducing a mathematically consistent decomposition of the polymer conformation tensor, into the so-called structure tensor. Our approach bypasses an inherent problem in the standard arithmetic decomposition, namely, the fluctuating conformation tensor fields may not be positive definite and hence, do not retain their physical meaning. Using well-established results in matrix analysis, the space of positive definite matrices is transformed into a Riemannian manifold by defining and constructing a geodesic via the inner product on its tangent space. This geodesic is utilized to define three scalar invariants of the structure tensor, which do not suffer from the caveats of the regular invariants (such as trace and determinant) of the polymer conformation tensor. First, we consider the problem of formulating perturbative expansions of the structure tensor using the geodesic, which is consistent with the Riemannian manifold geometry. A constraint on the maximum time, during which the evolution of the perturbative solution can be well approximated by linear theory along the Euclidean manifold, is found. Finally, direct numerical simulations of the viscoelastic sub-diffusive channel flows (where the stress-constitutive law is obtained via coarse-graining the polymer relaxation spectrum at finer scale, Chauhan et. al., Phys. Fluids, DOI: 10.1063/5.0174598 (2023)), underscore the advantage of using these invariants in effectively quantifying the macrostructures.
\end{abstract}

%
\vspace{2pc}
\noindent{\it Keywords}: anomalous diffusion, Caputo derivative, Riemannian manifold, structure-preserving groups, rheology

%
%
%

\section{Introduction}\label{sec:Intro}
The subject of anomalous diffusion has received tremendous attention over the last half-century, ranging from physics~\cite{Goychuk2021}, biology~\cite{Lai2009} to quantitative finance~\cite{Coffey2004}. Some of the most enigmatic and profoundly significant experimental results are better rationalized within the viscoelastic sub-diffusive approach in random environments such as the cytosol and the plasma membrane of biological cells~\cite{Rubenstein2003}, crowded complex fluids and polymer solutions~\cite{Levine2001}, dense colloidal suspensions~\cite{Kremer1990}, single-file diffusion in colloidal systems~\cite{Kou2004} as well as in atherosclerotic blood vessels~\cite{Fogelson2015}. A notably abnormal feature in viscoelastic sub-diffusive flows is the presence of temporally stable, non-homogeneous regions of high viscosity at moderate to high fluid inertia (or the so-called spatiotemporal macrostructures). For example, Riley~\cite{Riley1988} reported an elasticity induced flow stabilization of viscoelastic fluids coated over {\color{black} compliant} surfaces at a fairly high Reynolds number ($Re \sim 4000$). In a separate study involving ethanol gel fuels, elastic stabilization at a high shear rate was attributed due to an abnormally high second normal stress difference~\cite{Nandagopalan2018}. Viscoelastic flow stabilization at higher values of $Re$, in tapered microchannels, was explained due to the presence of wall effects~\cite{Zarabadi2019}. In another in vitro study, a biofilm deacidification created a non-homogeneous environment for molecular diffusion, leading to a `subdiffusive effect' with hindered flow rates~\cite{Zarabadi2018}. In summary, while the in silico studies of the classical channel flows indicate the appearance of temporal instability for Reynolds number as low as $Re \sim 50$~\cite{Khalid2021}, temporal stability and `structure formation' for viscoelastic sub-diffusive flows is only recognized in experimental realizations, until now. One objective of this work is to highlight the potential of fractional calculus to effectively capture the formation of macrostructures in viscoelastic sub-diffusive flows.

Previous approaches to analyze the polymer dynamics in dilute solutions have been to utilize the statistics of polymer forces~\cite{Dubief2005} and torques~\cite{Kim2013}. However, a more appropriate quantity to probe the polymer deformation history, is the conformation tensor, $\mathcal{C}$, a second order positive definite tensor which is obtained by averaging, over all molecular realizations, the dyad formed by the polymer end-to-end vector~\cite{Bird1987}. The trace of $\mathcal{C}$ (denoted from here onwards as $\tr \mathcal{C}$), is commonly used in literature to analyze $\mathcal{C}$ since (i) it is equal to the sum of its principal stretches and is, therefore, a measure of the polymer deformation~\cite{Sureshkumar1997}, (ii) it is proportional to the elastic energy in purely Hookean constitutive models of the polymers~\cite{Beris1994}. However, $\tr \mathcal{C}$ is not a sufficiently complete descriptor of polymer deformation. For example, Berris~\cite{Beris1994} found that the mean of $\tr \mathcal{C}$ can increase with increasing elasticity without a commensurate effect on the mean velocity profile. This behavior arises because the mean stress deficit is not a function of any of the normal components of $\mathcal{C}$. This example highlights the importance of simultaneously considering all of the components of $\mathcal{C}$ in order to arrive at a complete picture of the polymer deformation and its effect on the velocity field. The fluctuating conformation tensor, $\mathcal{C'}$ (obtained by subtracting the mean conformation tensor, $\mathcal{\overline{C}}$ from instantaneous tensor $\mathcal{C}$) and its moments, provide one method to obtain relevant higher-order statistical descriptions of $\mathcal{C}$. However, this fluctuating tensor is not guaranteed to be physically realizable since (i) whenever $\tr \mathcal{C'} \le 0$, this implies negative material deformation and this tensor loses positive-definiteness, and (ii) equally probable states of contraction ($\tr \mathcal{C} \in (0, 1)$) and expansion ($\tr \mathcal{C} \in (1, \infty)$) would be described by fluctuations with very different magnitudes. A more appealing way to evaluate fluctuations in $\mathcal{C}$ is to use $\log \mathcal{C}$ because the logarithm of a positive definite matrix is a symmetric matrix and the set of symmetric matrices form a vector space~\cite{Bhatia2015}. While $\log \mathcal{C}$ has been an object of interest in some studies of viscoelastic flows~\cite{Fattal2004}, two additional difficulties arise in using $\log \mathcal{C}$. First, the mean value of $\mathcal{C}$, or $\overline{\mathcal{C}}$, is not equal to $\rme^{\log \overline{\mathcal{C}}}$ implying that the effect of the polymer stress on the mean momentum balance requires all statistical moments of $\log \mathcal{C}$, even when the polymer stress is a linear function of $\mathcal{C}$. A second difficulty is that, in general, $\rme^{\log \overline{\mathcal{C}} + (\log \mathcal{C'})} \ne \rme^{\log \overline{\mathcal{C}}} \cdot \rme^{(\log \mathcal{C'})}$, which implies that there is no way to associate $\log \mathcal{C'}$ with a physical polymer deformation. Thus, this article is dedicated to the development of an alternate tensor from the polymer conformation tensor $\mathcal{C}$ as well as a formal way to visualize this new tensor, for sub-diffusive flows.

Although the mathematical results outlined in this work are well-established results in advanced matrix analysis textbooks~\cite{Lang2001, Bhatia2015}, to the author's best knowledge, they have not been used to evaluate the hydrodynamics of sub-diffusive flows. In this work, we aim to (i) derive an appropriate tensor (or the so called `structure tensor') which describes the polymer deformation in a physically realizable manner, (ii) derive appropriate scalar measures associated with the structure tensor, and (iii) corroborate our theory developed in aim-(i) and (ii) through regular perturbation analysis and fully nonlinear simulations. The paper is organized as follows. Our mathematical model along with the assumptions are delineated in~\sref{sec:model}. Equations describing the dynamics of the structure tensor is presented in~\sref{subsec:structure}. \Sref{subsec:invariants} outlines the main result, namely, the description of three scalar invariants of the structure tensor via the development of a geodesic on the Riemannian manifold. The weakly nonlinear perturbation analysis and the direct numerical simulations (DNS) are outlined in~\sref{sec:perturb} and~\sref{sec:results}, respectively. The conclusions follow in \sref{sec:conclusions}. Finally, a detailed derivation of the perturbed solution comprising the initial conditions for the numerical simulations is listed in~\ref{appA}.

\section{Mathematical Model}\label{sec:model}
In this section, we outline the model governing the incompressible, sub-diffusive dynamics of a planar (2D) viscoelastic channel flow for polymer melts. {\color{black} In an earlier study~\cite{Chauhan2022}, the authors derived the model by coarse-graining the polymer relaxation spectrum at finer scale, which resulted in a (time) fractional order, non-linear stress constitutive equations in the continuum limit.} Using the following scales for non-dimensionalizing the governing equations: the height of the channel $H$ for length, the timescale $T$ corresponding to maximum base flow velocity, $\mathcal{U}_0$ (i.~e., $T = (H / \mathcal{U}_0)^{1/\alpha}$) for time and $\rho \mathcal{U}^2_0$ for stresses (where $\rho, \mathcal{U}_0$ are the density and the velocity scale, respectively), we summarize the model in streamfunction-vorticity formulation as follows,
\begin{numparts} 
\label{eqn:FullSystem}
\begin{eqnarray}
&Re \left[ \frac{\partial^\alpha \Omega}{\partial t^\alpha} + {\mathbf{v}} \cdot \nabla \Omega \right] = \nu \nabla^2 \Omega + \frac{(1-\nu)}{We} \nabla \times \nabla \cdot \mathcal{C}, \label{eqn:Momentum} \\
&\nabla^2 \psi = -\Omega, \label{eqn:Poisson} \\
&\frac{\partial^\alpha \mathcal{C}}{\partial t^\alpha} + {\mathbf{v}}\cdot \nabla \mathcal{C} - (\nabla {\mathbf{v}})^T\mathcal{C} -\mathcal{C} \nabla {\mathbf{v}} = \frac{{\mathbf{I}} - \mathcal{C}}{We}, \label{eqn:ExtraStress}
\end{eqnarray}
\end{numparts}
where $\partial^\alpha f({\mathbf x}, t)/\partial t^\alpha$ denotes the Caputo fractional derivative of order $\alpha$~\cite{Podlubny1999} with respect to $t$ defined by
\begin{eqnarray}
\frac{\partial^\alpha f({\mathbf x}, t)}{\partial t^\alpha} \!= \!\frac{1}{\Gamma(1\!-\!\alpha)} \int^t_0 \!\! \frac{d t'}{(t-t')^\alpha} \frac{\partial f({\mathbf x}, {\color{black}t'})}{\partial t'}, \, 0\! <\! \alpha \!<\! 1,
\label{eqn:method2}
\end{eqnarray}
and the operators $\nabla(\cdot)$ and $\nabla^2(\cdot)$ in equation~\eref{eqn:Momentum}, are (integer order) gradient and Laplacian operators in $\mathbb{R}^2$. The variables $t, {\color{black} \psi}, {\mathbf v} = (u, v) = \left(\partial {\color{black} \psi}/\partial y, - \partial {\color{black} \psi}/\partial x\right),$ $\Omega = \nabla \times {\mathbf v}$, denote time, streamfunction, velocity and vorticity, respectively. The parameters $\eta_{\rm{s}}, \eta_{\rm{p}}, \eta_0 = \eta_{\rm{s}} + \eta_{\rm{p}}$ and $\nu = \eta_{\rm{s}} / \eta_0$ are the solvent viscosity, the polymeric contribution to the shear viscosity, the total viscosity and the viscous contribution to the total viscosity of the fluid, respectively. The dimensionless groups characterizing inertia and elasticity are Reynolds number, $Re = \rho \mathcal{U}_0 H/\eta_0$, and Weissenberg number, $We = \lambda^\alpha \mathcal{U}_0/H$, respectively. The parameter, {\color{black} $\lambda^\alpha$}, is the polymer relaxation time. Note that {\color{black}stress constitutive} equation~\eref{eqn:ExtraStress} represents the fractional version of the regular Oldroyd-B model for viscoelastic fluids~\cite{Sircar2019}.

From the perspective of continuum mechanics, $\mathcal{C}$, is the Finger tensor associated with polymer deformation~\cite{Beris1994}, such that
\begin{eqnarray}
\mathcal{C} = {\mathbf{F}} {\mathbf{F}}^T,
\label{eqn:CF}
\end{eqnarray}
where ${\mathbf{F}} = {\mathbf{F}}(t)$ is the instantaneous deformation gradient tensor. If the spatial coordinates in the micro-structure are given by ${\mathbf{r}} = {\mathbf{r}}({\mathbf{r}}_0, t)$ where ${\mathbf{r}}_0$ are the coordinates at equilibrium, then ${\mathbf{F}} = \partial {\mathbf{r}} / \partial {\mathbf{r}}_0$. In other words, a vector $d {\mathbf{r}}_0$ deforms to $d {\mathbf{r}} = {\mathbf{F}} d {\mathbf{r}}_0$ under the deformation, {\color{black}$\mathbf{F}$}.

\subsection{Dynamics of structure tensor}\label{subsec:structure}
The caveats in the conformation tensor outlined in~\sref{sec:Intro} (namely, the loss of positive definiteness in arithmetic compilation of fluctuations and unequal measure for equally probable states representing contraction and expansion) enforces us to adopt a different framework to capture polymer deformation in sub-diffusive flows. We begin by denoting the general linear group of degree $n$, which is the set of all $n \times n$ invertible matrices, as ${\mathbf{GL}}_n$. 
\begin{definition}
Define the structure-preserving group action of ${\mathbf{GL}}_n$ on a set ${\mathbf{V}}_n \subseteq \mathbb{R}^{n \times n}$ as,
\begin{displaymath}
[{\mathbf{B}}]_{\mathbf{A}} \equiv {\mathbf{A}} {\mathbf{B}} {\mathbf{A}}^T,
\end{displaymath} \label{def:SPGroup}
where ${\mathbf{A}} \in {\mathbf{GL}}_n$ and ${\mathbf{B}} \in {\mathbf{V}}_n$. 
\label{def:SPGA}
\end{definition}
Using definition~\ref{def:SPGroup}, we find that equation~\eref{eqn:CF} reduces to
\begin{eqnarray}
\mathcal{C} = [{\mathbf{I}}]_{\mathbf{F}}.
\label{eqn:S1}
\end{eqnarray}
Let $\overline{\mathcal{C}}$ be the mean conformation tensor (or the conformation tensor associated with the flow at equilibrium, refer~\sref{sec:results} for an example), then we assume that $\mathcal{C}$ is similar to $\overline{\mathcal{C}}$ under the group action (definition~\ref{def:SPGroup}) for any rotation matrix ${\mathbf{S}} \in {\mathbf{SO}}_{\color{black} n} \subseteq {\mathbf{GL}}_n$, where ${\mathbf{SO}}_n$ represents the $n \times n$ special orthogonal group of rotation matrices.

Similarly, define $\overline{\mathbf{F}} \in {\mathbf{GL}}_{\color{black} n}$ as the deformation gradient tensor associated with the mean configuration such that,
\begin{eqnarray}
\overline{\mathcal{C}} = \overline{{\mathbf{F}}} \overline{{\mathbf{F}}}^T.
\label{eqn:S2}
\end{eqnarray}
We remark that ${\overline{\mathbf{F}}}$ is non-unique since it can be represented as
\begin{eqnarray}
\overline{{\mathbf{F}}} = \overline{\mathcal{C}}^{1/2} {\mathbf{S}},
\label{eqn:S3}
\end{eqnarray}
for any ${\mathbf{S}} \in {\mathbf{SO}}_3$. $\overline{\mathcal{C}}^{1/2}$ is the unique matrix square-root, found exclusively in terms of $\overline{\mathcal{C}}$ and its invariants (i.~e., its trace and determinants) using an application of the representation theorem~\cite{Bhatia2015}. Since all we require is that the $\det \overline{{\mathbf{F}}} > 0$ {\color{black} (in order to maintain the positive definiteness of $\overline{\mathcal{C}}$)}, we choose ${\mathbf{S}} = {\mathbf{I}}$ in equation~\eref{eqn:S3}. 

Given $\overline{{\mathbf{F}}}$ satisfying equation~\eref{eqn:S3}, we can decompose the instantaneous deformation gradient tensor ${\mathbf{F}}$ and $\overline{\mathbf{F}}$ by considering successive transformations on the vector $d {\mathbf{r}}_0$ as,
\begin{eqnarray}
d {\mathbf{r}} = {\mathbf{F}} d {\mathbf{r}}_0 = \overline{{\mathbf{F}}} \mathcal{L} d {\mathbf{r}}_0,
\label{eqn:S4}
\end{eqnarray}
where $\mathcal{L} = \overline{{\mathbf{F}}}^{-1} {\mathbf{F}}$ is the tensor describing fluctuations away from the mean configuration, {\color{black} denoted as fluctuating deformation gradient tensor}. Alternatively, substituting ${\mathbf{F}} = \overline{{\mathbf{F}}} \mathcal{L}$ in equation~\eref{eqn:CF} and utilizing definition~\ref{def:SPGroup}, we arrive at the following definition, 
\begin{definition}[Structure tensor]
Define $\mathcal{G}$ such that
\begin{displaymath}
\mathcal{C} = \overline{{\mathbf{F}}} \mathcal{G} \overline{{\mathbf{F}}}^T = [\mathcal{G}]_{\overline{{\mathbf{F}}}},
\end{displaymath}
where $\mathcal{G} = \mathcal{L}\mathcal{L}^T$. 
\label{def:ST}
\end{definition}

The {\color{black} fluctuating conformation tensor}, $\mathcal{C'} (= \mathcal{C}-\overline{\mathcal{C}})$ is related to the structure tensor as follows,
\begin{eqnarray}
\mathcal{C'} = [\mathcal{G} - {\mathbf{I}}]_{\overline{{\mathbf{F}}}}.
\label{eqn:S6}
\end{eqnarray}

Using definition~\ref{def:ST} in equation~\eref{eqn:ExtraStress} and pre multiplying (post multiplying) the resultant equation by $\overline{{\mathbf{F}}}^{-1}$ ($\overline{{\mathbf{F}}}^{-T}$) and noting that $\mathcal{G}$ is a symmetric tensor, we arrive at the following equation governing the dynamics of structure tensor, 
\begin{eqnarray}
\frac{\partial^\alpha \mathcal{G}}{\partial t^\alpha} + {\mathbf{v}} \cdot \nabla \mathcal{G} = \mathcal{G} \mathfrak{F}({\mathbf{v}}) + {\mathfrak{F}({\mathbf{v}})}^T \mathcal{G} - \mathcal{M}
\label{eqn:ST}
\end{eqnarray}
which replaces equation~\eref{eqn:ExtraStress} in the viscoelastic sub-diffusive model (equation~\eref{eqn:FullSystem}). The functions, \\
$\mathfrak{F}({\mathbf{v}}) = {\overline {\mathbf{F}}}^T \nabla {\mathbf{v}} \overline {\mathbf{F}}^{-T} - {\left({\overline {\mathbf{F}}}^{-1} \left({\mathbf{v}} \cdot \nabla \right) \overline {\mathbf{F}} \right)}^T$, and $\mathcal{M} = \frac{1}{We} \left(\mathcal{G} - {\left({\overline {\mathbf{F}}}^{T} {\overline {\mathbf{F}}}\right)}^{-1}\right)$.

\subsection{Main results: scalar invariants via a non-euclidean geodesic}\label{subsec:invariants}
The tensorial nature of $\mathcal{G}$ renders the quantification of the fluctuating {\color{black} conformation tensor}, a difficult task. By utilizing $\tr \mathcal{G}$, Berris~\cite{Beris1994} made an initial attempt to characterize polymer deformation in the Oldroyd-B model, by defining a `elastic potential energy'. Elastic energy was an insufficient descriptor in characterizing polymer deformation due to (a) its dependence on the choice of the particular constitutive model, and (b) elastic energy was found to be the same for a family of conformation tensors with identical trace but variable determinant. We instead evolve an approach to characterize deformation using the inherent structure of the tensor $\mathcal{G}$.

Any scalar characterization of $\mathcal{G}$ can be naively developed as a function of its three principle invariants, i.~e., trace, dyadic product of eigenvalues and determinant. However, even for simple cases (such as the isotropic case) the invariants are bounded between $0$ and $1$ ($1$ and $\infty$) for compression with respect to $\overline{\mathcal{C}}$ (expansion with respect to $\overline{\mathcal{C}}$). This asymmetric characterization is undesirable. Further, the statistical moments of the invariants vary over several orders of magnitude, rendering these moments as uninformative predictors of polymer stretching. The above-mentioned problems arise because the set of $n \times n$ positive definite matrices (denoted with $\mathbf{PS}_n$ in subsequent discussion) do not form a vector space and thus the euclidean notion of translation and distances are irrelevant. Instead, we exploit the Riemannian structure of $\mathbf{PS}_n$ to formulate alternative scalar measures of $\mathcal{G}$.

$\mathbf{PS}_n$ is a Hilbert space where we can define an inner product given by $\ainnerproduct{\mathbf{A}}{\mathbf{B}}_{\mathbf{X}}  =  \tr \left( \mathbf{X}^{-1} \cdot \mathbf{A}^T \cdot \mathbf{X}^{-1} \cdot \mathbf{B} \right)$, and a corresponding induced norm $\|{\mathbf{A}}\|_{\mathbf{X}} = \left( \tr \left( \mathbf{X}^{-1} \mathbf{A}^T \mathbf{X}^{-1} \mathbf{A} \right) \right)^{1/2}$, where $\mathbf{X} \in \mathbf{PS}_n$. Furthermore, since $\mathbf{PS}_n$ is an open subset of the space of $n \times n$ real-valued matrices, it is a differentiable manifold. {\color{black} Using a simple argument, it can be shown that} the tangent space at every point {\color{black}in $\mathbf{PS}_n$} is the space of symmetric matrices. However, $\mathbf{PS}_n$ can be {\color{black} shown} to be a Riemannian manifold {\color{black}with a} geodesic {\color{black}which is obtained} via the {\color{black}same} inner product {\color{black}(defined above)} on the tangent space at every point. The next set of results form the requisite machinery to formulate this geodesic which will be needed to define the scalar invariants of $\mathcal{G}$. 

Consider a parametrized curve on $\mathbf{PS}_n$ connecting points $\mathbf{X}, \mathbf{Y}  \in \mathbf{PS}_n$. That is $P: [0,1] \to \mathbf{PS}_n$ with $P(0) = \mathbf{X}$ and $P(1) = \mathbf{Y}$. The distance, in the sense of the Riemannian metric, traversed on the manifold along the curve $P = P(r)$ is given by,
\begin{displaymath}
  \ell_P(r) = \int\limits_0^r \norm{\frac{dP(r')}{dr'}}_{P(r')} dr'. 
\end{displaymath}
$\ell_P$ is {\color{black}invariant under affine transformation, as shown in the next lemma.}

\begin{lemma}[Affine invariance~\cite{Bhatia2015}]
 For every positive definite matrix $\mathbf{A}$ and differentiable path $P$ on the Riemannian manifold of positive definite matrices, we have:
  \[
    \ell_{P} = \ell_{[P]_{\mathbf{A}}},
  \]
  where $[\cdot]_{\mathbf{A}}$ denotes an action under $\mathbf{A}$ of the form $\mathbf{A}^T P \mathbf{A}$ {\color{black}(see definition~\ref{def:SPGA} for details)}.
  \label{lem:AI}
\end{lemma}
\begin{proof}
  We use the definition of the norm $\|{\cdot}\|_{\mathbf{X}}$ as stated above and the commutativity of the trace of matrix product, to arrive at,  

\begin{eqnarray}
\fl{\color{black}\norm{d P_{[\mathbf{A}]}}_{P_{[\mathbf{A}]}}}&=&\norm{\left( \mathbf{A}^T P(r) \mathbf{A} \right)^{{\color{black} -1/2}} \left( \mathbf{A}^T P(r) \mathbf{A} \right)^{\color{black} '} \left( \mathbf{A}^T P(r) \mathbf{A} \right)^{-1/2}}_{\color{black} \mathbf{I}} \nonumber\\
\fl  &=& {\color{black} \left( \tr \left( \left( \mathbf{A}^T P(r) \mathbf{A} \right)^{-1} \left( \mathbf{A}^T P(r) \mathbf{A} \right)' \left( \mathbf{A}^T P(r) \mathbf{A} \right)^{-1} \left( \mathbf{A}^T P(r) \mathbf{A} \right)' \right) \right)^{1/2}} \nonumber \\
\fl  &=& \left( \tr \mathbf{A}^{-1} \left( P \right)^{-1}(r) P'(r) \left( P \right)^{-1}(r) {\color{black} P'(r)} \mathbf{A} \right)^{1/2} \nonumber \\
\fl  &=& \left( \tr \left( P \right)^{-1}(r) P'(r) \left( P \right)^{-1}(r) P'(r) \right)^{1/2} \nonumber \\
\fl  &=& \norm{\left( P \right)^{-1/2}(r) P'(r) \left( P \right)^{-1/2}(r)}_{\color{black} \mathbf{I}} {\color{black}= \norm{d P}_P}, \nonumber
\end{eqnarray}

{\color{black} where $^\prime$ denotes derivative with respect to the independent variable.} The proof is completed by integrating both sides of this equality over $r$ to obtain that,
\[
  \ell_P = \int\limits_0^r \norm{\frac{dP(r')}{dr'}}_{P(r')} dr' = {\color{black} \int\limits_0^r \norm{\frac{dP_{[\mathbf{A}]}(r')}{dr'}}_{P_{[\mathbf{A}]}(r')} dr'} = \ell_{[P]_{\mathbf{A}}}.
\]
\end{proof}
In the derivation above, we have used the cyclical property of the trace and the fact that an infinitesimal distance away from the point $\mathbf{X}$ on the manifold is given by $\norm{d \mathbf{X}}_{\mathbf{X}} = \norm{\mathbf{X}^{-1/2} d \mathbf{X} \mathbf{X}^{-1/2}}_{\mathbf{I}} = \left(\tr \left( \mathbf{X}^{-1} d \mathbf{X}\right)^2\right)^{1/2}$. Using lemma~\ref{lem:AI}, we can define $d(\mathbf{X}, \mathbf{Y})$, (or the geodesic distance between $\mathbf{X}$ and $\mathbf{Y}$) as the infimum of $\ell_P(1)$ over all possible curves $P$ connecting $\mathbf{X}$ and $\mathbf{Y}$,
\begin{definition}
\begin{eqnarray}
  d(\mathbf{X}, \mathbf{Y}) = \inf_{P} \{\ell_{P}(1) \, | \, P(r) \in \mathbf{PS}_n, P(0) = \mathbf{X}, P(1) = \mathbf{Y}. \} \nonumber
\end{eqnarray}
  \label{def:geodesic1}
\end{definition}
{\color{black}A corollary of lemma~\ref{lem:AI} is that $d(\mathbf{X}, \mathbf{Y}) = d([\mathbf{X}]_{\mathbf{A}}, [\mathbf{Y}]_{\mathbf{A}})$.} Note that the Hopf-Rinow theorem guarantees the existence and uniqueness of such a geodesic. The next set of three results allow the construction of this geodesic.

\begin{theorem}[Exponential metric increasing property~\cite{Lang2001}]
  For any two real symmetric matrices $\mathbf{X}$ and $\mathbf{Y}$, we have that:
\begin{eqnarray}
    \norm{\mathbf{X} - \mathbf{Y}}_{\mathbf{I}} \le d(\rme^{\mathbf{X}}, \rme^{\mathbf{Y}}),
    \label{eqn:T1}
 \end{eqnarray}
  where $\rme^{\mathbf{X}}, \rme^{\mathbf{Y}}$ are positive definite matrices.
  \label{theorem:thm1}
\end{theorem}
\begin{proof}
In order to demonstrate the proof, we wish to show the following inequality,
  \begin{eqnarray} \label{eqn:IMEI}
    \norm{\mathbf{X} - \mathbf{Y}}_{\mathbf{I}} \le \norm{\rme^{-\mathbf{A}/2} \left( D \rme^{\mathbf{A}} (\mathbf{X} - \mathbf{Y}) \right) \rme^{-\mathbf{A}/2}}_{\mathbf{I}}
  \end{eqnarray}
  where $\mathbf{A}$ is any real symmetric matrix, $\rme^{-\mathbf{A}}$ is the exponential map evaluated at the point $\mathbf{A}$ in the Riemannian manifold of symmetric matrices, and $D \rme^{\mathbf{A}} (\mathbf{X} - \mathbf{Y})$ is the derivative of the exponential map at the point $\mathbf{A}$ evaluated on the matrix $\mathbf{X} - \mathbf{Y}$ and defined as follows,
 \begin{definition}
\begin{eqnarray}
     D \rme^{\mathbf{H}} (\mathbf{K}) := \lim_{t \to 0} \frac{\rme^{\mathbf{H} + t \mathbf{K}} - \rme^{\mathbf{H}}}{t}. \nonumber
\end{eqnarray}
 \label{def:Dexp}
\end{definition}
for any matrices $\mathbf{H}$ and $\mathbf{K}$. {\color{black} The inequality~\eref{eqn:T1} follows from the inequality~\eref{eqn:IMEI} as follows. Let $\mathbf{H}(t)$ be any path joining symmetric matrices $\mathbf{X}$ and $\mathbf{Y}$, then $\rme^{\mathbf{H}}$ is the path joining $\rme^{\mathbf{X}}$ and $\rme^{\mathbf{Y}}$. Let $\chi = \rme^{\mathbf{H}}$ then $\chi' = D \rme^{\mathbf{H}} ({\mathbf{H}}'(t))$. The length of this path is given by,}
\begin{eqnarray}
{\color{black}\ell_\chi = \int_0^1 \norm{d \chi}_\chi} &=& \int_0^1 \norm{\chi^{-1/2} d\chi \chi^{-1/2}}_{\mathbf{I}} dr \nonumber\\
&=& \int_0^1 \norm{\rme^{-\mathbf{H}/2} D \rme^{\mathbf{H}} \mathbf{H}'(r) \rme^{-\mathbf{H}/2}}_{\mathbf{I}} dr \nonumber\\
&\ge& \int_0^1 \norm{\mathbf{H}'(r)}_{\mathbf{I}} dr \nonumber\\
&\ge& \norm{\mathbf{X} - \mathbf{Y}}_{\mathbf{I}},
\label{Bhatia_Cor6.1.3}
\end{eqnarray}
%
where the last inequality appears because $\norm{\mathbf{X} - \mathbf{Y}}_{\mathbf{I}}$ is the length of this path in the Euclidean space of symmetric matrices. But $d(\rme^{\mathbf{X}}, \rme^{\mathbf{Y}}) = \inf \ell_\chi \ge \norm{\mathbf{X} - \mathbf{Y}}_{\mathbf{I}}.$ (from Definition~\ref{def:geodesic1}).
%
 
Thus, in order to prove inequality~\eref{eqn:IMEI} we shall equivalently show for a real symmetric matrix, $\mathbf{B} = (\mathbf{X} - \mathbf{Y})$ that,
\begin{eqnarray} \label{eqn:IMEI1}
    \norm{\mathbf{B}}_{\mathbf{I}} \le \norm{\rme^{-\mathbf{A}/2} \left( D \rme^{\mathbf{A}} (\mathbf{B}) \right) \rme^{-\mathbf{A}/2}}_{\mathbf{I}}.
\end{eqnarray}
{\color{black} Choosing an orthonormal basis in which $\mathbf{A} =$diag$(\lambda_1, \ldots, \lambda_n)$ and deploying the spectral decomposition formula (abridged from equation~(2.40) in~\cite{Bhatia2015}), we have that,}
\begin{eqnarray}
{\color{black}D \rme^{\mathbf{A}}(\mathbf{B}) = \frac{\rme^{\lambda_i}-\rme^{\lambda_j}}{\lambda_i - \lambda_j} b_{ij},} \nonumber
\end{eqnarray}
{\color{black}where $b_{ij}$ is the $i, j$ entry of the matrix, $\mathbf{B}$. Similarly, the $i, j$ entry of the matrix, $\rme^{-\mathbf{A}/2} \left( D \rme^{\mathbf{A}} (\mathbf{B}) \right) \rme^{-\mathbf{A}/2}$, is,}
\begin{eqnarray}
{\color{black}\frac{\sinh (\lambda_i - \lambda_j) / 2}{(\lambda_i - \lambda_j) / 2} b_{ij}.} \nonumber
\end{eqnarray}
Since $(\sinh x) / x \ge 1$ for all real $x$, the inequality~\eref{eqn:IMEI1} follows.
\end{proof}

We note that the equality in equation~\eref{eqn:T1} is achieved when $\mathbf{X}$ and $\mathbf{Y}$ commute, and in this {\color{black} special} case, we can parametrize the geodesic, as outlined in the next result.
\begin{proposition}\label{proposition:Prop1}
  Let $\mathbf{X} = \rme^{\mathcal{X}}$ and $\mathbf{Y} = \rme^{\mathcal{Y}}$ be positive definite matrices such that $\mathbf{X} \mathbf{Y} = \mathbf{Y} \mathbf{X}$. Then, the exponential function maps the line segment~\cite{Bhatia2015},
\[
    (1 - r) \mathcal{X} + r \mathcal{Y}, \quad 0 \le r \le 1,
\]
  in the Euclidean space of symmetric matrices to the geodesic between $\mathbf{X}$ and $\mathbf{Y}$ on the Riemannian manifold of positive definite matrices and
\begin{eqnarray}
    d(\mathbf{X}, \mathbf{Y}) = \norm{\mathcal{X} - \mathcal{Y}}_{\mathbf{I}}.
    \label{eqn:T7}
\end{eqnarray}
\end{proposition}

\begin{proof}
  It is enough to show that the path given by,
\begin{eqnarray}
    {\color{black}\gamma(r)} = \rme^{\left( (1 - r) \mathcal{X} + r \mathcal{Y} \right)}, \quad 0 \le r \le 1,
   \label{eqn:T8}  
\end{eqnarray}
is the unique path of shortest length joining $\mathbf{X}$ and $\mathbf{Y}$ in the space of symmetric matrices. Adopting the following parametrization of the path, ${\color{black}\gamma(r)} = \mathbf{X}^{1 - r} \mathbf{Y}^r$ as well as the commutativity of $\mathbf{X}$ and $\mathbf{Y}$, we have,
 \begin{eqnarray}
    {\color{black}\gamma'(r) =  \left( \mathcal{Y} - \mathcal{X} \right) \gamma(r).}
\end{eqnarray}
{\color{black} The length of this path is given by (see equation~\eref{Bhatia_Cor6.1.3} above),}
\begin{eqnarray}
{\color{black} \ell_\gamma = \int_0^1 \norm{\gamma^{-1/2} d \gamma \gamma^{-1/2}}_{\mathbf{I}} dr = \norm{\mathcal{X} - \mathcal{Y}}_{\mathbf{I}}.}
\label{Bhatia_Prop6.1.5}
\end{eqnarray}
{\color{black} Theorem~\ref{theorem:thm1} says that $\gamma$ is the shortest path. All that remains to show is that the path $\gamma$ under consideration is unique. Suppose $\tilde{\gamma}$ is another path that joins $\mathbf{X}$ and $\mathbf{Y}$ and has the same length as that of $\gamma$. Then $\log \tilde{\gamma}$ is a path joining $\mathcal{X}$ and $\mathcal{Y}$, and by equation~\eref{Bhatia_Cor6.1.3} this path has minimum length $\norm{\mathcal{X} - \mathcal{Y}}_{\mathbf{I}}$. However, in a Euclidean space, the straight line segment is the unique shortest path between two points. Hence, the result follows.}
%
\end{proof}

Finally, using the affine-invariance property (lemma~\ref{lem:AI}) of the Riemannian metric and noting that $\mathbf{I}$ commutes with every element of $\mathbf{PS}_n$, we arrive at the following {\color{black} general} result.
\begin{theorem}
  Let $\mathbf{X}$ and $\mathbf{Y}$ be positive definite matrices. There exists a unique geodesic $\mathbf{X} \#_r \mathbf{Y}$ on the Riemannian manifold of positive definite matrices that joins $\mathbf{X}$ and $\mathbf{Y}$ with the following parametrization~\cite{Bhatia2015},
\begin{eqnarray}
    \mathbf{X} \#_r \mathbf{Y} = \mathbf{X}^{1/2} \left( \mathbf{X}^{-1/2} \mathbf{Y} \mathbf{X}^{-1/2} \right)^r \mathbf{X}^{1/2},
\label{eqn:T10}
\end{eqnarray}
which is natural in the sense that,
\begin{eqnarray}
    d(\mathbf{X}, \mathbf{X} \#_r \mathbf{Y}) = r \, d(\mathbf{X}, \mathbf{Y}),
\label{eqn:T11}
\end{eqnarray}
  for each $r \in {\color{black}[0,\,1]}$. Furthermore, we have,
  \begin{eqnarray}
  \fl   d(\mathbf{X}, \mathbf{Y}) = \norm{ \log \left( \mathbf{X}^{-1/2} \mathbf{Y} X^{1/2} \right)}_{\mathbf{I}} = \left[ \sum\limits_{i = 1}^3 \left( \log \sigma_i \left( \mathbf{X}^{-1} \mathbf{Y} \right) \right)^2 \right]^{1/2},
\label{eqn:T12}
  \end{eqnarray}
where $\sigma_i$ are the eigenvalues of the matrix $\mathbf{X}^{-1} \mathbf{Y}$.
\label{theorem:thm2}
\end{theorem}
\begin{proof}
 Clearly, the matrices $\mathbf{I}$ and $\mathbf{X}^{-1/2} \mathbf{Y} \mathbf{X}^{-1/2}$ commute. Hence, the geodesic joining these {\color{black}two} points is naturally parameterized as:
  \begin{eqnarray}
    P_0(r) = \left( \mathbf{X}^{-1/2} \mathbf{Y} \mathbf{X}^{-1/2} \right)^r.
    \label{eqn:T13}
  \end{eqnarray}
  Applying the isometry $\mathbf{X}^{1/2} \mathbf{Y} \mathbf{X}^{1/2}$, we obtain the path,
  \begin{eqnarray}
    P(r) = \mathbf{X}^{1/2} \left( P_0(r) \right) \mathbf{X}^{1/2} = \mathbf{X}^{1/2} \left( \mathbf{X}^{-1/2} \mathbf{Y} \mathbf{X}^{-1/2} \right)^r \mathbf{X}^{1/2},
    \label{eqn:T14}
  \end{eqnarray}
  joining the points $\mathbf{X}^{1/2} \mathbf{I} \mathbf{X}^{1/2} = \mathbf{X}$ and $\mathbf{X}^{1/2} \mathbf{X}^{-1/2} \mathbf{Y} \mathbf{X}^{-1/2} \mathbf{X}^{1/2} = \mathbf{Y}$. Because $\mathbf{X}^{1/2} \mathbf{Y} \mathbf{X}^{1/2}$ is an isometry, the path~\eref{eqn:T14} is a geodesic joining $\mathbf{X}$ and $\mathbf{Y}$. Thus, equality~\eref{eqn:T10} follows. Next, in order to prove the equality~\eref{eqn:T11} we have that,
  \begin{eqnarray}
    \fl d\left(\mathbf{X}, \mathbf{X} \#_r \mathbf{Y}\right)
    &=& d \left(\mathbf{X}, \mathbf{X}^{1/2} \left( \mathbf{X}^{-1/2} \mathbf{Y} \mathbf{X}^{-1/2} \right)^r \mathbf{X}^{1/2}\right) \nonumber \\ 
    \fl &=& {\color {black}d \left(\mathbf{X}^{-1/2} \mathbf{X} \mathbf{X}^{-1/2}, \mathbf{X}^{-1/2} \mathbf{X}^{1/2} \left( \mathbf{X}^{-1/2} \mathbf{Y} \mathbf{X}^{-1/2}\right)^r \mathbf{X}^{1/2} \mathbf{X}^{-1/2} \right)} \nonumber \\
    \fl &=& {\color {black} d \left(\mathbf{I}, \left( \mathbf{X}^{-1/2} \mathbf{Y} \mathbf{X}^{-1/2}\right)^r \right) }\nonumber \\
   \fl  &=& {\color {black}\inf \int\limits_0^1 \norm{\log(\mathbf{I}) - \log \left( \mathbf{X}^{-1/2} \mathbf{Y} \mathbf{X}^{-1/2}\right)^r } dr} \nonumber \\
   \fl  &=& {\color {black} r \inf \int\limits_0^1 \norm{ \log \left( \mathbf{X}^{-1/2} \mathbf{Y} \mathbf{X}^{-1/2}\right) }} = r \, d \left(\mathbf{X}, \mathbf{Y}\right),
       \label{eqn:T15}
  \end{eqnarray}
  where the last equality arises as a consequence of proposition~\ref{proposition:Prop1}, since,
  \begin{eqnarray}
    d \left({\color{black}\mathbf{X}, \mathbf{Y}}\right)
    &=& d\left( \mathbf{I}, \mathbf{X}^{-1/2} \mathbf{Y} \mathbf{X}^{-1/2} \right) \nonumber \\
    &=& \norm{\log \mathbf{I} - \log \left( \mathbf{X}^{-1/2} \mathbf{Y} \mathbf{X}^{-1/2} \right)}_{\mathbf{I}} \nonumber \\
    &=& \norm{\log \left( \mathbf{X}^{-1/2} \mathbf{Y} \mathbf{X}^{-1/2} \right)}_{\mathbf{I}}.
       \label{eqn:T16}
  \end{eqnarray}

  Finally, from the definition of the Riemannian norm, and basic linear algebra, we get that,
  \begin{eqnarray}
    d(\mathbf{X}, \mathbf{Y}) = \left[ \sum\limits_{i = 1}^3 \left( \log \sigma_i \left( \mathbf{X}^{-1} \mathbf{Y} \right) \right)^2 \right]^{1/2}, 
        \label{eqn:T17}
  \end{eqnarray}
  where $\sigma_i$ are the eigenvalues of the matrix $\mathbf{X}^{-1} \mathbf{Y}$.
\end{proof}
We are now ready to introduce the scalar measures which can be used to quantify the {\color{black} structure} tensor, $\mathcal{G}$. First, let us denote the matrix logarithm of $\mathcal{G}$ as $\mathcal{L}_{\mathcal{G}}$ (i.~e., $\mathcal{G} = \rme^{\mathcal{L}_{\mathcal{G}}}$). This matrix logarithm exists, is unique (since $\mathcal{G}$ is positive definite) and has eigenvalues which are the logarithm of the eigenvalues of $\mathcal{G}$.

\subsubsection{Scalar invariants 1: volume ratio}\label{subsubsec:delta1}
Let $\sigma_i(\mathcal{G})$ ($i\,=\,1, 2, 3$) be the eigenvalues of $\mathcal{G}$. Define the first scalar invariant as the volume ratio of fluctuations, $\delta_1$, as
\begin{definition}
\begin{eqnarray}
\delta_1 = \tr \mathcal{L}_{\mathcal{G}} = \log \det \mathcal{G} = \log \left( \frac{\det \mathcal{C}}{\det \overline{\mathcal{C}}} \right). \nonumber
\end{eqnarray}
\label{def:delta1}
\end{definition}
when $\delta_1 = 0$, the mean and the instantaneous conformation tensors have the same volume and when $\delta_1$ is negative (positive), the instantaneous conformation tensor has smaller (larger) volume than the mean volume. 

\subsubsection{Scalar invariants 2: shortest distance from mean}\label{subsubsec:delta2}
When $\mathcal{C} = \overline{\mathcal{C}}$, we have $\mathcal{G} = I$. When $\mathcal{C} \ne \overline{\mathcal{C}}$, we wish to consider the shortest path between $I$ and $\mathcal{G}$ as a measure of the magnitude of fluctuations. Using equation~\eref{eqn:T17}, we consider the squared geodesic distance related with this path,
\begin{definition}
\begin{eqnarray}
\delta_2 = \tr \mathcal{L}^2_{\mathcal{G}} = d^2({\mathbf I}, \mathcal{G}) = \sum^3_{i=1} \left(\log \sigma_i\right)^2. \nonumber
\end{eqnarray}
\label{def:delta2}
\end{definition}
Using equation~\eref{eqn:T17}, we can verify that $d^2({\mathbf I}, \mathcal{G}) = d^2({\mathbf I}, \mathcal{G}^{-1})$, which implies that this squared geodesic treats both expansions and compressions identically. The affine invariance property (lemma~\ref{lem:AI}) ensures that $d({\mathbf I}, \mathcal{G}) = d([{\mathbf I}]_{{\mathbf A}}, [\mathcal{G}]_{{\mathbf A}})$. With ${\mathbf A} = \overline{{\mathbf F}}$, we obtain,
\begin{eqnarray}
d^2({\mathbf I}, \mathcal{G}) = d^2(\overline{\mathcal{C}}, [\mathcal{G}]_{\overline{F}}) = d^2(\overline{\mathcal{C}}, \mathcal{C}) = d^2(\overline{\mathcal{C}}^{-1}, \mathcal{C}^{-1}).
\label{eqn:delta2}
\end{eqnarray}
The last equality in equation~\eref{eqn:delta2} exhibits the fact that the metric introduced in equation~\eref{eqn:T17}, handles expansions and compressions on equal terms, unlike the regular Euclidean metric (or the Frobenius norm). 

\subsubsection{Scalar invariants 3: anisotropy index}\label{subsubsec:delta3}
Following Hameduddin~\cite{Hameduddin2018}, we define the anisotropy index, $\delta_3$ as the squared geodesic distance between $\mathcal{G}$ and the closest isotropic tensor,
\begin{definition}
\begin{eqnarray}
\delta_3 = \inf d^2(a{\mathbf I}, \mathcal{G}) = \inf_a \tr(\mathcal{G} - (\log a){\mathbf I})^2. \nonumber
\end{eqnarray}
\label{def:delta3}
\end{definition}
Through differentiation, we find that $a = \left(\prod\limits_{i=1}^{3} \sigma_i\right)^{1/3} = \left(\det \mathcal{G}\right)^{1/3}$ is the minimizing stationary point, which implies that,
\begin{eqnarray}
\delta_3 = d^2(\left(\det \mathcal{G}\right)^{1/3}{\mathbf I}, \mathcal{G}) = \delta_2 - \frac{1}{3} \delta^2_1.
\end{eqnarray}
Notice that $\delta_3 = 0$ if and only if $ \delta^2_1 = 3 \delta_2$, in which case $\mathcal{G}$ reduces to an isotropic tensor.

Finally, we surmise that equations~\eref{eqn:Momentum},~\eref{eqn:Poisson} and~\eref{eqn:ST} alongwith definitions~\ref{def:ST},~\ref{def:delta1},~\ref{def:delta2} and~\ref{def:delta3} form the complete set of equations governing the dynamics of viscoelastic sub-diffusive flows. 

\section{Perturbative expansion for weakly nonlinear deformation}\label{sec:perturb}
A weakly nonlinear expansion up to the $K$th power of the velocity field is given by,
\begin{eqnarray}
  \Omega &= \overline{\Omega} + \sum\limits_{k = 1}^K \epsilon^k \Omega_k, \nonumber \\
  \psi &= \overline{\psi} + \sum\limits_{k = 1}^K \epsilon^k \psi_k,
  \label{eqn:Pert1}
\end{eqnarray}
where the superscript, $\overline{(\,\,\,) }$, denote the mean values and $\Omega_k(\mathbf{x}, t), \,\, \psi_k,\,\,\,\left(k \in [1,\,K]\right)$ are the perturbed vorticities and stream-functions of the $k$th-order, respectively. A similar expansion for $\mathcal{G}$ is inappropriate because it is positive definite and there is no {\em a priori} guarantee on this property with regular arithmetic expansion. Instead, we adopt the geometric expansion by multiplicatively decomposing the fluctuating deformation gradient tensor into $K$ separate components,
\begin{eqnarray}
\mathcal{L}_{\rm{Pert}} = \mathcal{L}_1 \mathcal{L}_2 \ldots \mathcal{L}_k
\label{eqn:Pert2}
\end{eqnarray}

Using definition~\ref{def:ST} as well as the matrix logarithm of $\mathcal{G}$ (=$\mathcal{L}_{\mathcal{G}}$), we can express the perturbed structure tensor at $k$th-stage of decomposition as,
\begin{eqnarray}
{\color{black} \mathcal{G}_k^{\epsilon^k}} = \mathcal{L}_k \mathcal{L}^T_k  = \rme^{\epsilon^k \mathcal{L}_{\mathcal{G}_k}},
  \label{eqn:Pert3}
\end{eqnarray}
where $\mathcal{L}_{\mathcal{G}_0} = {\mathbf 0}$. From equation~\eref{eqn:Pert2}, we can associate a perturbed tensor, 
\begin{eqnarray}
\mathcal{L}_k = \rme^{\epsilon^k \mathcal{L}_{\mathcal{G}_k}/2}. 
  \label{eqn:Pert4}
\end{eqnarray}
At this stage, we remark that although we assume that each $\mathcal{L}_k$ is positive definite, the product of positive definite tensors, $\mathcal{L}_{\rm{Pert}}$ (equation~\eref{eqn:Pert2}) is not necessarily positive definite. However, since $\mathcal{L}_{\rm{Pert}} \mathcal{L}^T_{\rm{Pert}} = \mathcal{L}\mathcal{L}^T = \mathcal{G}$, we can show via a polar decomposition that $\mathcal{L}_{\rm{Pert}} = \mathcal{L} \mathbf{R}$ for some rotation tensor, $\mathbf{R}$. Substituting equation~\eref{eqn:Pert4} and~\eref{eqn:Pert2} in equation~\eref{eqn:Pert3}, we arrive at the necessary expansion,
%
\begin{eqnarray}
\fl \mathcal{G} &=& \rme^{\epsilon \mathcal{L}_{\mathcal{G}_1}/2} \ldots \rme^{\epsilon^{K-1} \mathcal{L}_{\mathcal{G}_{K-1}}/2} \rme^{\epsilon^K \mathcal{L}_{\mathcal{G}_K}} \rme^{\epsilon^{K-1} \mathcal{L}_{\mathcal{G}_{K-1}}/2} \ldots \rme^{\epsilon \mathcal{L}_{\mathcal{G}_1}/2},\nonumber \\
\fl &=& \mathbf{I} \!+\! \epsilon \mathcal{L}_{\mathcal{G}_1} \!+\! \epsilon^2 \left( \frac{\mathcal{L}^2_{\mathcal{G}_1}}{2} + \mathcal{L}_{\mathcal{G}_2}\right)\! +\! \epsilon^3 \!\! \left( \frac{\mathcal{L}^3_{\mathcal{G}_1}}{6} + \frac{1}{2}(\mathcal{L}_{\mathcal{G}_1} \mathcal{L}_{\mathcal{G}_2} + (\mathcal{L}_{\mathcal{G}_1} \mathcal{L}_{\mathcal{G}_2})^T) + \mathcal{L}_{\mathcal{G}_3} \right) \!+\! \ldots,
  \label{eqn:Pert5}
\end{eqnarray} 
%
where the second equality in equation~\eref{eqn:Pert5} makes use of the matrix exponential, $\rme^{\epsilon^k \mathcal{L}_{\mathcal{G}_k}} = \sum\limits_{q=0}^\infty \epsilon^{kq} \mathcal{L}^q_{\mathcal{G}_k} / q!$. Substituting expansions~\eref{eqn:Pert1} and~\eref{eqn:Pert5} in equations~\eref{eqn:Momentum},~\eref{eqn:Poisson} and~\eref{eqn:ST}, we arrive at the $\mathcal{O}(\epsilon)$ model equations,
\numparts
\label{eqn:Pert6}
\begin{eqnarray}
\fl & Re\left[ \frac{\partial^\alpha \Omega_1}{\partial t^\alpha} + \overline{\mathbf{v}} \cdot \nabla \Omega_1 + \mathbf{v}_1 \cdot \nabla \overline{\Omega} \right] = \nu \nabla^2 \Omega_1 + \frac{1 - \nu}{W_e} \nabla \times \nabla \cdot \left(\overline{\mathbf{F}} \mathcal{L}_{\mathcal{G}_1} \overline{\mathbf{F}}^T \right), \label{eqn:Pert6a} \\
\fl &\nabla^2 \psi_1 = - \Omega_1, \label{eqn:Pert6b} \\
\fl &\frac{\partial^\alpha}{\partial t^\alpha} \mathcal{L}_{\mathcal{G}_1} + \overline{\mathbf{v}} \cdot \nabla \mathcal{L}_{\mathcal{G}_1} = 2 \rm{sym}\left( \mathfrak{F}(\mathbf{v}_1) + \mathcal{L}_{\mathcal{G}_1} \mathfrak{F}(\overline{\mathbf{v}}) \right) - \frac{1}{W_e}\mathcal{L}_{\mathcal{G}_1}, \label{eqn:Pert6c}
\end{eqnarray}
\endnumparts
as well as the $\mathcal{O}(\epsilon^2)$ model equations,
\numparts\label{eqn:Pert7}
\begin{eqnarray}
\fl &Re\left[ \frac{\partial^\alpha \Omega_2}{\partial t^\alpha} + \overline{\mathbf{v}} \cdot \nabla \Omega_2 + \mathbf{v}_2 \cdot \nabla \overline{\Omega} + \mathbf{v}_1 \cdot \nabla \Omega_1 \right] \nonumber\\
\fl & \qquad \qquad \qquad = \nu \nabla^2 \Omega_2  + \frac{1 - \nu}{W_e} \nabla \times \nabla \cdot \left(\overline{\mathbf{F}} \left( \frac{\mathcal{L}_{\mathcal{G}_1}^2}{2} + \mathcal{L}_{\mathcal{G}_2} \right) \overline{\mathbf{F}}^T \right), \label{eqn:Pert7a}  \\
\fl &\nabla^2 \psi_2 = - \Omega_2, \label{eqn:Pert7b} \\
\fl &\frac{\partial^\alpha \mathcal{L}_{\mathcal{G}_2}}{\partial t^\alpha} +\ \overline{\mathbf{v}} \cdot \nabla \mathcal{L}_{\mathcal{G}_2} + \mathbf{v}_1 \cdot \nabla \mathcal{L}_{\mathcal{G}_1} \nonumber \\
\fl & \qquad \qquad \qquad = {\color{black} 2 \rm{sym}\left( \mathfrak{F}(\mathbf{v}_2) + \mathcal{L}_{\mathcal{G}_2} \mathfrak{F}(\overline{\mathbf{v}}) \right) + \frac{\mathcal{L}_{\mathcal{G}_1}^2}{2 W_e} - \frac{\mathcal{L}_{\mathcal{G}_2}}{W_e}} -\mathcal{L}_{\mathcal{G}_1} \rm{sym}\left(\mathfrak{F}(\overline{\mathbf{v}})\right) \nonumber \\
\fl & \qquad \qquad \qquad \mathcal{L}_{\mathcal{G}_1} + \mathcal{L}_{\mathcal{G}_1} \rm{asym} \left(\mathfrak{F}(\mathbf{v}_1)\right) {\color{black} - \rm{asym} \left(\mathfrak{F}(\mathbf{v}_1)\right) \mathcal{L}_{\mathcal{G}_1} }, \label{eqn:Pert7c} 
\end{eqnarray}
\endnumparts
where sym$(\mathbf{A}) = (\mathbf{A}+\mathbf{A}^T)/2$ {\color{black} , asym$(\mathbf{A}) = (\mathbf{A}-\mathbf{A}^T)/2$} and $\mathbf{v}_i = (u_i, v_i) = \left(\partial \psi_i/\partial y, -\partial \psi_i/\partial x\right)$.
\subsection{Linear perturbations}\label{subsec:linear}
As an illustration, we highlight the case of linear perturbative solutions for the 2D viscoelastic channel flow for polymer melts. A rectilinear coordinate system is used with $x, y$ denoting the channel flow direction and the transverse direction, respectively. The origin of this coordinate system is chosen at the left end of the lower wall of the channel. The size of the domain is chosen to be $(x, y) \in \Gamma = [0,\,5] \times [0,\,1]$. The mean flow is assumed to be a plane Poiseuille flow with its variation entirely in the transverse direction, namely,
\begin{eqnarray}
{\mathbf U}_0 = (y - y^2) {\mathbf e_x},
\label{eqn:L1}
\end{eqnarray}
where ${\mathbf e_x}$ is the unit vector along x-direction. The mean flow, ${\mathbf U}_0$, defines the mean vorticity, $\overline{\Omega} = 2y-1$, and the mean stream-function, $\overline{\psi} = \frac{y^2}{2} - \frac{y^3}{3}$. In this case, the initial conditions for the perturbed solution can be constructed via the superposition of the mean flow and the instability mode, as follows,
\begin{eqnarray}
     \Omega|_{t=0} &\approx \overline{\Omega} + \epsilon \Omega_1 = \overline{\Omega} + \epsilon \mathcal{R}\{\widetilde{\Omega}(y)|_{t=0} \rme^{\rmi k x} \}, \nonumber \\
             \psi|_{t=0} &\approx \overline{\psi} + \epsilon \psi_1 = \overline{\psi} + \epsilon \mathcal{R}\{\widetilde{\psi}(y)|_{t=0} \rme^{\rmi k x} \}, \nonumber\\
\mathcal{G}|_{t=0} &\approx \mathbf{I} + \epsilon \mathcal{L}_{\mathcal{G}_1}  = \mathbf{I} + \epsilon \mathcal{R}\{\widetilde{\mathcal{L}_{\color{black}{\mathcal{G}}}}(y)|_{t=0} \rme^{\rmi k x} \},
\label{eqn:L2}
\end{eqnarray}
where $(\Omega_1, \psi_1,\mathcal{L}_{\mathcal{G}_1})$ are the perturbations that are Fourier transformed in the $x$-direction. $\mathcal{R}\{ \}$ denotes the real part of the complex valued function. The equations governing the initial conditions~\eref{eqn:L2} are listed in~\ref{appA}, whose solution is unique {\color{black} upto} an integration constant for the stream-function perturbation (refer ~\fref{fig1}b for the solution).

An initial condition comprising of a small amplitude unstable mode will initially grow exponentially, as predicted by the linear theory. However, nonlinear effects eventually become significant since otherwise, the conformation tensor losses positive definiteness. Hence, our interest to study linear perturbations is to find an estimate of the maximum time during which the perturbed solution can be well approximated by the linear theory, i.~e., along the Euclidean manifold.

Consider an initial condition which is a perturbed base flow, $\mathcal{G}|_{t=0} = \mathbf{I} + \epsilon \mathcal{L}_{\mathcal{G}}$. If we assume that the perturbed mode grows according to linear theory for some time and $\mathcal{G}$ evolves along Euclidean lines with growth rate $\omega > 0$, then $\mathcal{G}(t) = \mathbf{I} + \epsilon \mathcal{L}_{\mathcal{G}} \rme^{\omega t}$. Suppose $\mathcal{L}_{\mathcal{G}}$ is not zero and is harmonic in the spatial direction, then $\mathcal{L}_{\mathcal{G}}$ has a strictly negative eigenvalue somewhere in the domain. This is because a harmonic perturbation leads to regions where the polymers are much more compressed than the maximum expansion, in a volumetric sense, since positive and negative additive perturbations to the mean conformation tensor, $\overline{\mathcal{C}}$, with equal magnitude are not of equal magnitude with respect to the natural norm on $\mathbf{PS}_3$. For positive-definiteness of $\mathcal{G}$, we require the eigenvalues, $1 + \epsilon \sigma_i(\mathcal{L}_{\mathcal{G}})\rme^{\omega t} > 0, \,\, (i = 1, 2, 3)$. Whenever $\sigma_i(\mathcal{L}_{\mathcal{G}}) < 0$, the dynamics induces a curvature on the evolution of the perturbed mode along $\mathbf{PS}_3$ before the time, $t_m$, when the eigenvalue of $\mathcal{G}$ crosses zero. Hence, at the point of crossing zero, this time is given by,
\begin{eqnarray}
\omega t_m = -\left( \log \epsilon + \log \max|\sigma_i(\mathcal{L}_{\mathcal{G}})|\right),
\label{eqn:L3}
\end{eqnarray}
where $\max|\sigma_i(\mathcal{L}_{\mathcal{G}})|$ is the magnitude of the largest negative eigenvalue in the domain. Readers are referred to a previous work, on the instability of viscoelastic sub-diffusive channel flows, by the authors~\cite{Chauhan2022} to gain an insight on the procedure for finding, $\sigma_i(\mathcal{L}_{\mathcal{G}})$. Equation~\eref{eqn:L3} serves as a guide for selecting the initial perturbation amplitude, $\epsilon$, based on the time, $t_m$, i~e., by reducing $\epsilon$ we can arbitrarily increase $t_m$ to a desired value.

\begin{figure*}[htbp]
\centering
\includegraphics[width=0.475\linewidth, height=0.4\linewidth]{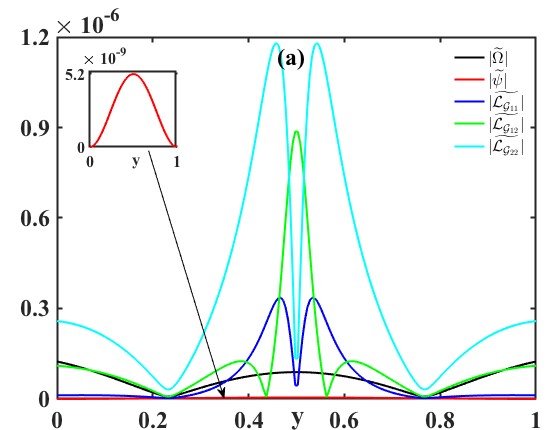}
\includegraphics[width=0.495\linewidth, height=0.375\linewidth]{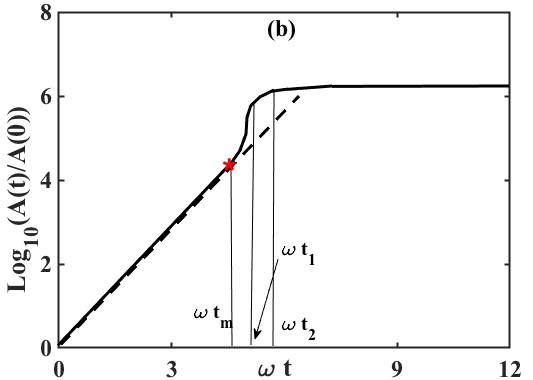}
\caption{(a) Solution to the linearized system of equations (\ref{appA}) subject to the boundary values, $(\widetilde{\psi}(y),\widetilde{\psi'}(y)) = (0, 0)$ at the rigid channel walls, $y = 0$ and $y = 1$, and (b) time evolution of $A$, as defined in equation~\eref{eqn:L4} for parameter values, $We=10.0, Re=70.0, \nu=0.3$ and $\alpha=0.5$. The solid line represents the nonlinear evolution, $A(t) / A(0)$ (equations~(\eref{eqn:Momentum},~\eref{eqn:Poisson} and~\eref{eqn:ST}) with initial conditions~\eref{eqn:L2}) while the dashed line represents the growth of the solution predicted by the linear theory~\eref{eqn:Pert6}. The asterisk ($\ast$) indicates the maximum time, $t_m$ as defined in equation~\eref{eqn:L3}.}
\label{fig1}
\end{figure*}
In order to compare the linear evolution of the unstable modes (equation~\eref{eqn:Pert6} with initial conditions~\eref{eqn:L2}) with the nonlinear modes (equations~(\eref{eqn:Momentum},~\eref{eqn:Poisson} and~\eref{eqn:ST}), the following quantity is utilized,
\begin{eqnarray}
A(t) = \frac{1}{\ell_x \ell_y} \int_\Gamma d^2(\mathbf{I}, \mathcal{G}) d \Gamma,
\label{eqn:L4}
\end{eqnarray}
which measures the perturbations away from the isotropic tensor, $\mathbf{I}$, in the volume-averaged sense~\cite{Hameduddin2018}. $(\ell_x, \ell_y)$ are the lengths of the domain in the flow direction and the transverse direction, respectively. The time evolution of the normalized function, $A(t) / A(0)$, for parameter values, $We=10.0, Re=70.0, \nu=0.3, \alpha=0.5, \epsilon=0.1, k = 0.01$ (refer equations~(\eref{eqn:Momentum},~\eref{eqn:Poisson} and~\eref{eqn:ST} as well as the initial conditions~\eref{eqn:L2}) is shown in ~\fref{fig1}b. We note that the evolution of $A(t) / A(0)$ matches with the one predicted by the linear theory~\eref{eqn:Pert6}, upto the maximum time, $\omega t_m \approx 4.6$ and then shows deviation from the linear growth in the form of an exponential growth upto time, $\omega t_1 \approx 5.2$, followed by an eventual saturation at $\omega t_2 \approx 6.0 $. While the initial exponential deviation can be explained due to the exponential form (refer equation~\eref{eqn:Pert5}) of the structure tensor, the mode saturation is the manifestation of nonlinear effects, which is absent in the linear theory. A detailed description of the nonlinear effects through numerical simulations is outlined in~\sref{sec:results}.
%
%
\section{Direct Numerical Simulations}\label{sec:results}
Next, the fully non-linear model~(\eref{eqn:Momentum},~\eref{eqn:Poisson} and~\eref{eqn:ST}) for planar, viscoelastic channel flow, subject to the initial conditions~\eref{eqn:L2}, is numerically investigated for two specific cases of the fractional order derivative, namely, the monomer diffusion in {\color{black} coarse-grained} Zimm chain solution ($\alpha = \case{2}{3}$)~\cite{Zimm1956} and {\color{black} coarse-grained} Rouse chain melts ($\alpha = \case{1}{2}$)~\cite{Rouse1953}. 

In order to imitate an infinitely long channel, periodic boundary conditions are assumed at the flow inlet and outlet. No-slip (i.~e., $u = v = 0$) and zero tangential conditions (i.~e., $\frac{\partial u}{\partial x} = \frac{\partial v}{\partial x} = 0$) are imposed on the lower wall ($y = 0$) and the upper wall ($y = 1.0$) of the channel, respectively. Further, incompressibility constraint provides an additional condition on the walls: $\frac{\partial v}{\partial y} = 0$. Since the flow is parallel to the channel walls, the walls may be treated as streamline. Thus, the streamfunction value, $\psi$, on the wall is set as a constant. That constant (which may be different on the lower and the upper wall) is found from the no-slip condition. Zero tangential condition imply that all tangential derivatives of streamfunction vanish on the wall. Thus, the boundary condition for vorticity is found from the Poisson equation~\eref{eqn:Poisson},
\begin{eqnarray}
\frac{\partial^2 \psi}{\partial y^2}|_{wall} = -\Omega_{wall}.
\label{eqn:OmegaBC}
\end{eqnarray}
Finally, the boundary conditions for the structure tensor is constructed from equation~\eref{eqn:ST}, coupled with the no-slip and zero tangential conditions, as follows,
\begin{eqnarray}
& \frac{\partial^\alpha \mathcal{G}_{11}}{\partial t^\alpha} + \frac{2}{\sqrt{d}} \frac{\partial^2 \psi}{\partial y^2} \left(F_1 F_2 \mathcal{G}_{11} + F_2^2 \mathcal{G}_{12}\right) + \frac{\mathcal{G}_{11}}{We} - \frac{F_2^2 + F_4^2}{d We} = 0, \nonumber \\
& \frac{\partial^\alpha \mathcal{G}_{12}}{\partial t^\alpha} - \frac{1}{\sqrt{d}} \frac{\partial^2 \psi}{\partial y^2} \left(F_1^2 \mathcal{G}_{11} - F_2^2 \mathcal{G}_{22}\right) + \frac{\mathcal{G}_{12}}{We} + \frac{F_1 F_2 + F_2 F_4}{d We} = 0, \nonumber \\
& \frac{\partial^\alpha \mathcal{G}_{22}}{\partial t^\alpha} - \frac{2}{\sqrt{d}} \frac{\partial^2 \psi}{\partial y^2} \left(F_1^2 \mathcal{G}_{12} + F_1 F_2 \mathcal{G}_{22}\right) + \frac{\mathcal{G}_{22}}{We} - \frac{F_1^2 + F_2^2}{d We} = 0,
\label{eqn:ExtraStressBC}
\end{eqnarray}
where the variables $d, F_i\,\, (i\,=\,1, 2, 4)$ are listed in~\ref{appA}.

The domain, $\Gamma = [0,\,5] \times [0,\,1]$, is discretized using $76 \times 51$ points such that the discrete points are equally spaced at $\Delta x = \frac{5}{75}$ and $\Delta y = \frac{1}{50}$, excluding the boundary points, where the periodic / Dirichlet boundary conditions are imposed in the flow direction / transverse direction, respectively. The  implicit-explicit time-adaptive, $\theta$-method~\cite{Chauhan2022b} is utilized for the numerical outcome, with the variable `$\theta$' is fixed at $\theta=1.0$. The minimum and the maximum values of the variable time-step are chosen as $\Delta t_{min}=10^{-3}$ and $\Delta t_{max}=1.6 \times 10^{-2}$, respectively. The Poisson equation~\eref{eqn:Poisson} is iteratively solved using the Gauss-Siedal iteration technique, without an explicit inversion of the coefficient matrix. {\color{black} Other algorithmic details may be found in a recently published work by the authors~\cite{Chauhan2022b}. Since the impact of elasticity and inertia on the flow rheology has been reported elsewhere~\cite{Chauhan2022b} and since our goal in this article is to highlight the advantages of the newly developed metrics over the traditional metrics, the flow-material parameters are fixed at $Re=70, We=10, \nu=0.3, \epsilon=0.1$.}

\subsection{{\color{black} Coarse-grained} Zimm's model}\label{subsec:Zimm}
The Zimm's model~\cite{Zimm1956} predicts the (`shear rate and polymer concentration independent') viscosity of the polymer solution by calculating the hydrodynamic interaction of flexible polymers (an idea which was originally proposed by Kirkwood~\cite{Kirkwood1954}) by approximating the chains using a bead-spring setup. 

The instantaneous principle invariant of the structure tensor, $\tr \mathcal{G}$ as well as the new invariants (refer~\sref{subsubsec:delta1}, \sref{subsubsec:delta2}, \sref{subsubsec:delta3}) for the Zimm's flow rheology are presented in ~\fref{fig2} (left column). The contours of the other principle invariants are qualitatively similar to those of the $\tr \mathcal{G}$ and are thus not shown here. Observe that the structures appearing in the Zimm's model are significantly smaller in magnitude than the Rouse model (\sref{subsec:Rouse}). Physically, the formation of these `spatiotemporal macrostructures' are associated with the entanglement of the polymer chains at microscale~\cite{Rubenstein2003}, leading to localized, non-homogeneous regions with higher viscosity.

~\Fref{fig2c} shows the logarithmic volume ratio, $\delta_1$. This quantity is qualitatively similar to the principle invariant, $\tr \mathcal{G}$, and hence we have a visual resemblance in ~\fref{fig2a},~\fref{fig2c}. However, we find that in ~\fref{fig2c}, we have predominantly negative values, indicating that the instantaneous volume is smaller than the volume of the mean conformation. Further, observe regions of very high values of $\delta_1$ interspersed with regions of very low values, especially near the wall. This observation is the result of the slow diffusion of polymers in sub-diffusive flows since there is no direct mechanism for smoothening out these `elastic shocks' in the tensor field. The measure, $\delta_1$ does not distinguish between volume-preserving deformations. For example, $\delta_1$ does not distinguish between $\mathcal{G}$ and $\det(\mathcal{G}) \mathcal{G}_1$, for any tensor $\mathcal{G}_1$ with a unit determinant. In particular, $\delta_1 = 0$, does not imply $\mathcal{C} = \overline{\mathcal{C}}$. In order to identify regions where the instantaneous polymer conformation equals the mean conformation and quantify the deviation when it is not, we use the squared geodesic distance away from the origin $(\mathbf{I})$ along the Riemannian manifold, $\delta_2$ (~\fref{fig2e}). ~\Fref{fig2e} indicates that the conformation tensor field is significantly far away from $\overline{\mathcal{C}}$, near the wall. This deviation of $\delta_2$, in the near wall region, can be explained via the `memory effect', previously observed in regular Oldroyd-B fluids~\cite{Sircar2019}. Finally, ~\fref{fig2g} shows the instantaneous contours of the anisotropy index, $\delta_3$. This index shows how close the shape of instantaneous conformation tensor is to the shape of the mean conformation tensor, irrespective of volumetric changes. The visual resemblance of $\delta_2$ and $\delta_3$ suggests that deformations to the mean conformation are largely anisotropic, near wall.
\begin{figure*}[htbp]
\centering
\subfloat[]{\label{fig2a}\includegraphics[width=0.49\linewidth, height=0.3\linewidth]{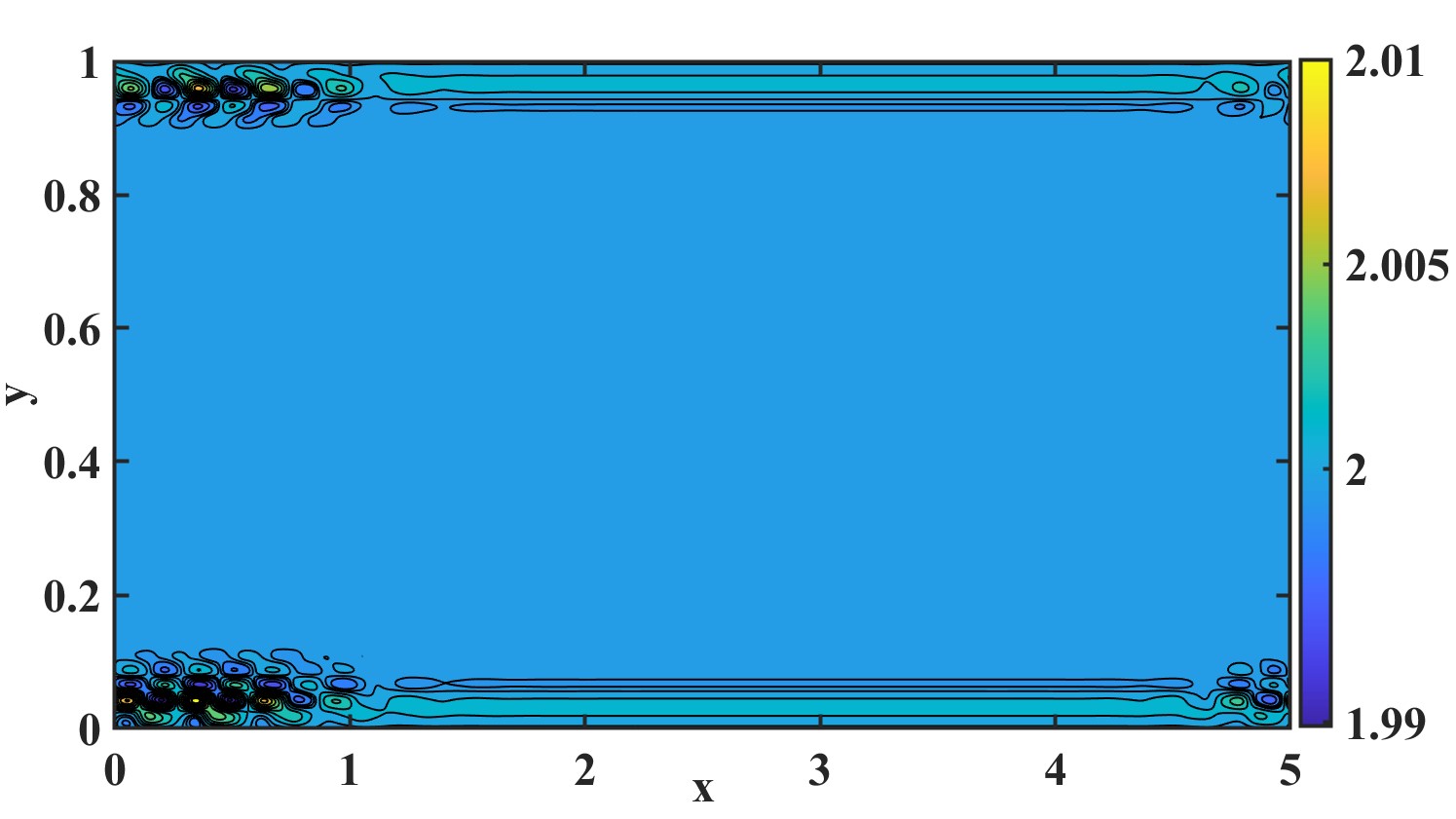}}
\hspace{0.1cm}
\subfloat[]{\label{fig2b}\includegraphics[width=0.49\linewidth, height=0.3\linewidth]{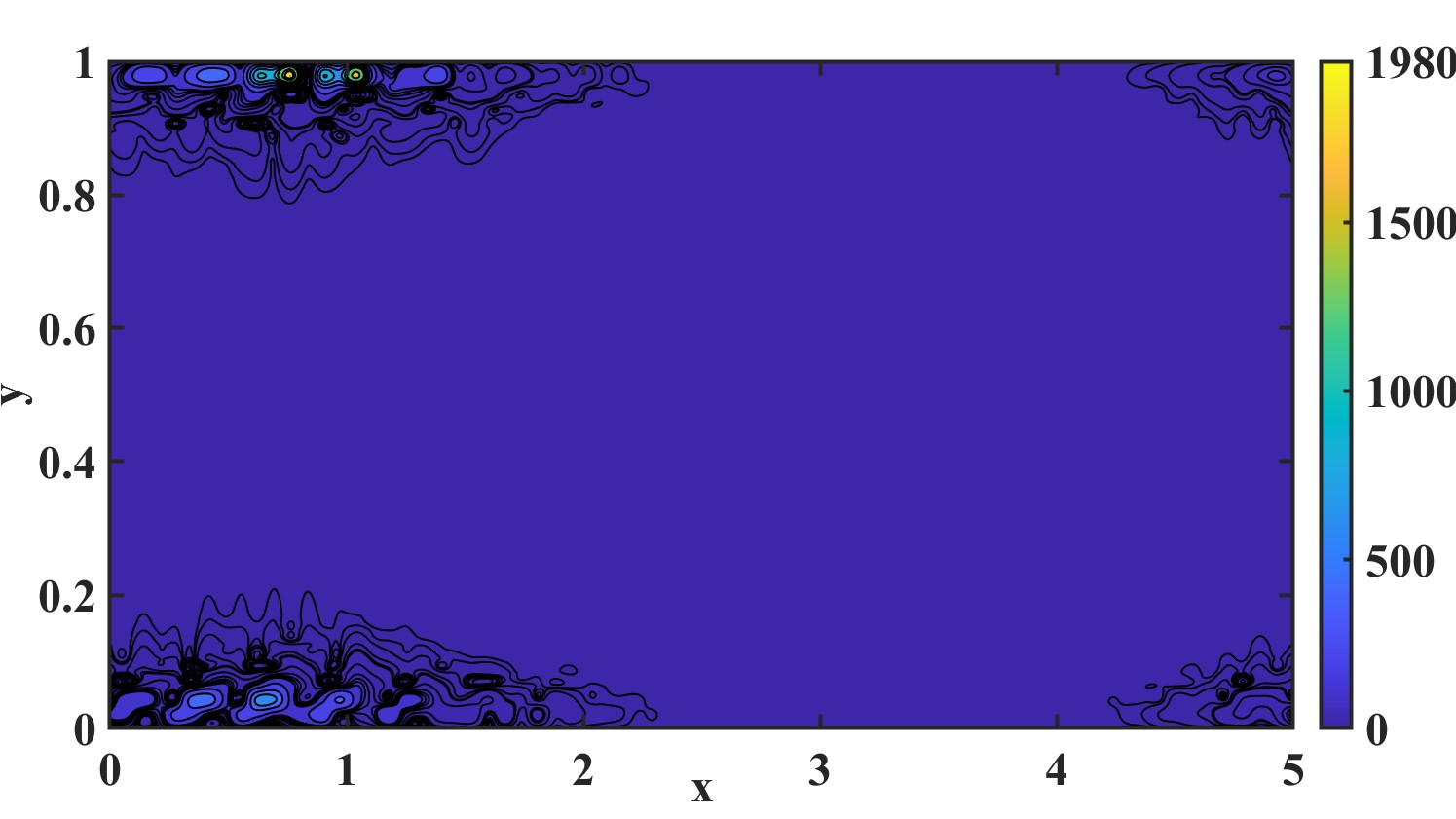}}\\
\subfloat[]{\label{fig2c}\includegraphics[width=0.49\linewidth, height=0.3\linewidth]{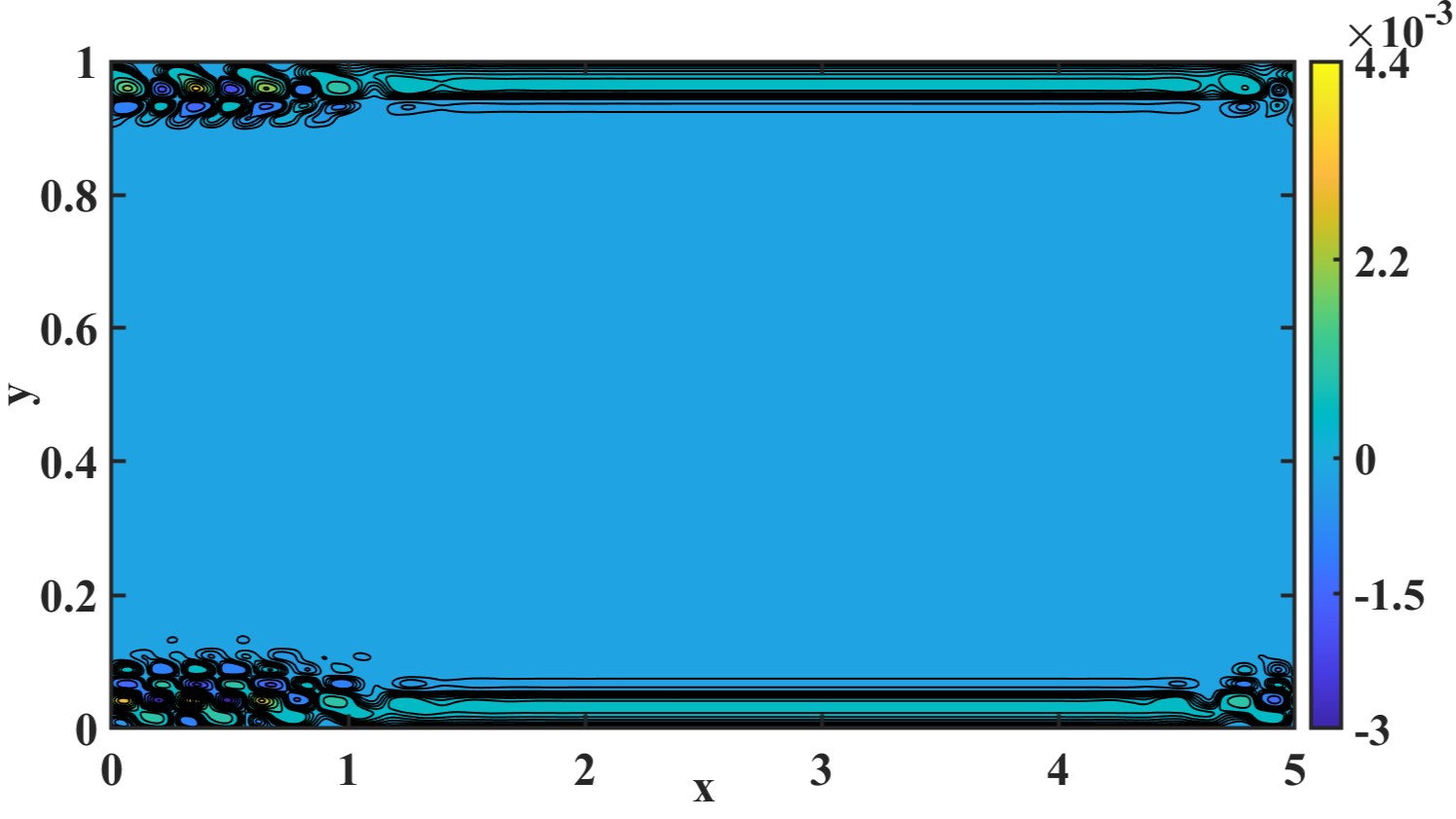}}
\hspace{0.1cm}
\subfloat[]{\label{fig2d}\includegraphics[width=0.49\linewidth, height=0.3\linewidth]{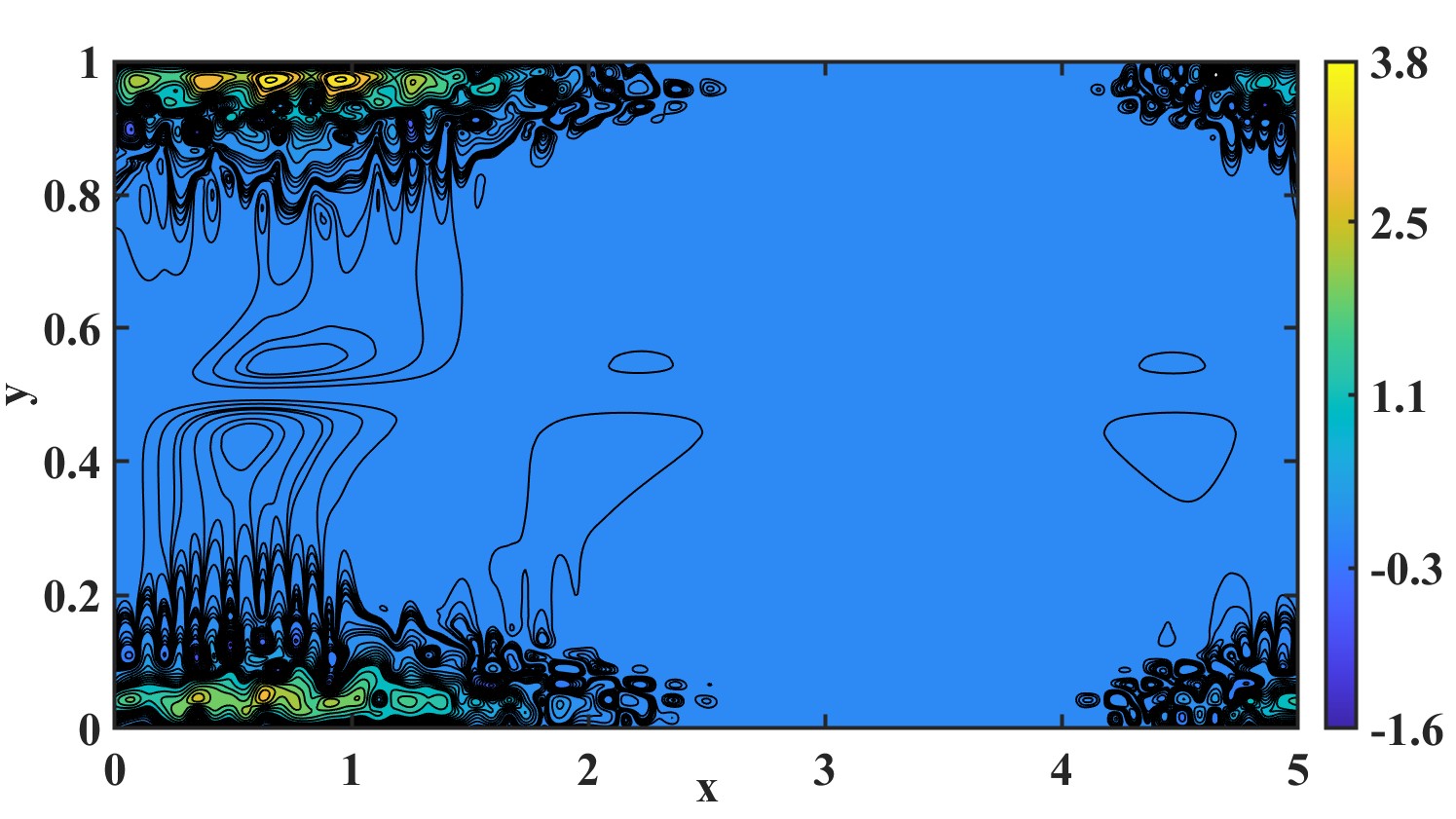}}\\
\subfloat[]{\label{fig2e}\includegraphics[width=0.49\linewidth, height=0.3\linewidth]{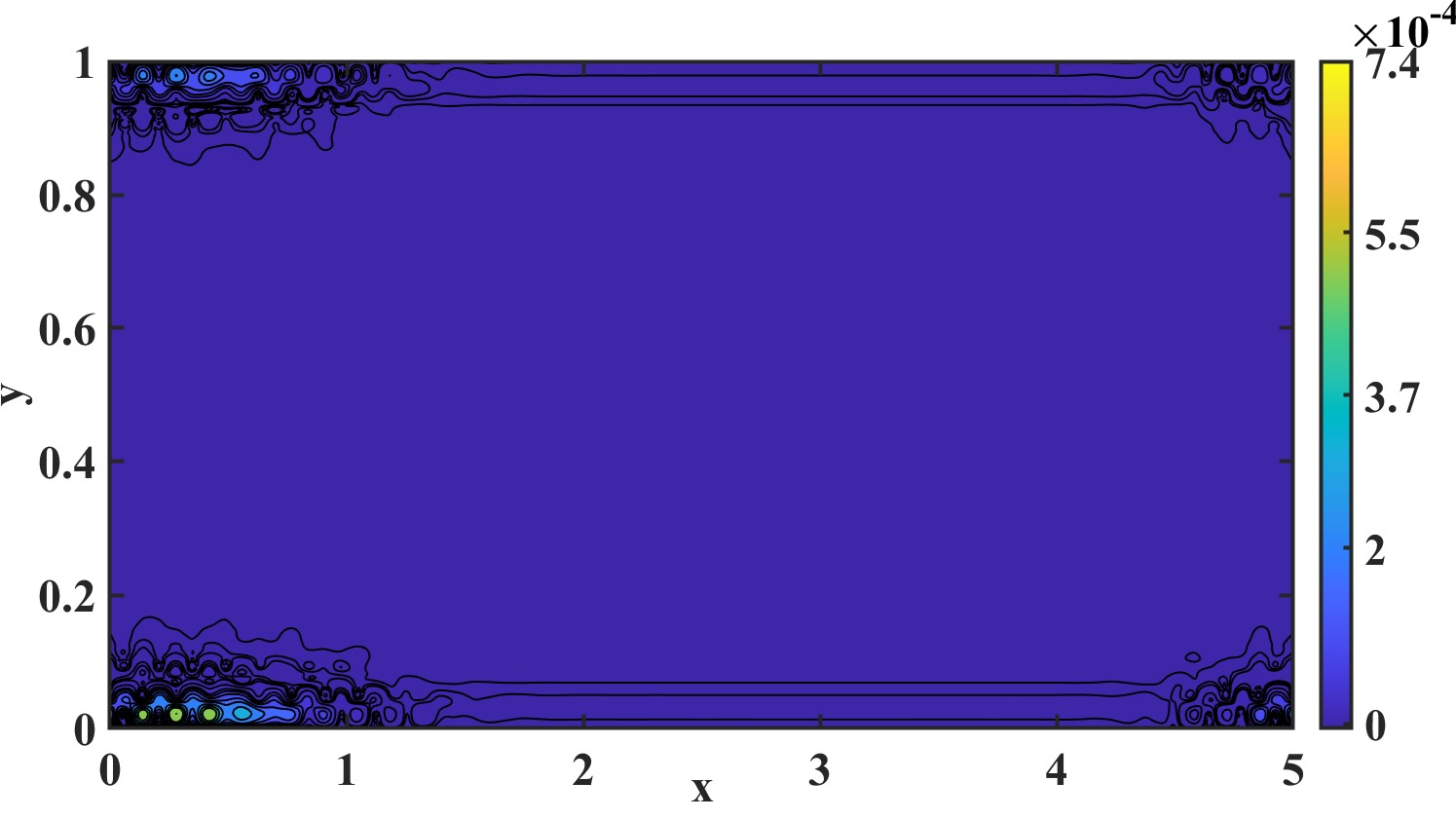}}
\hspace{0.1cm}
\subfloat[]{\label{fig2f}\includegraphics[width=0.49\linewidth, height=0.3\linewidth]{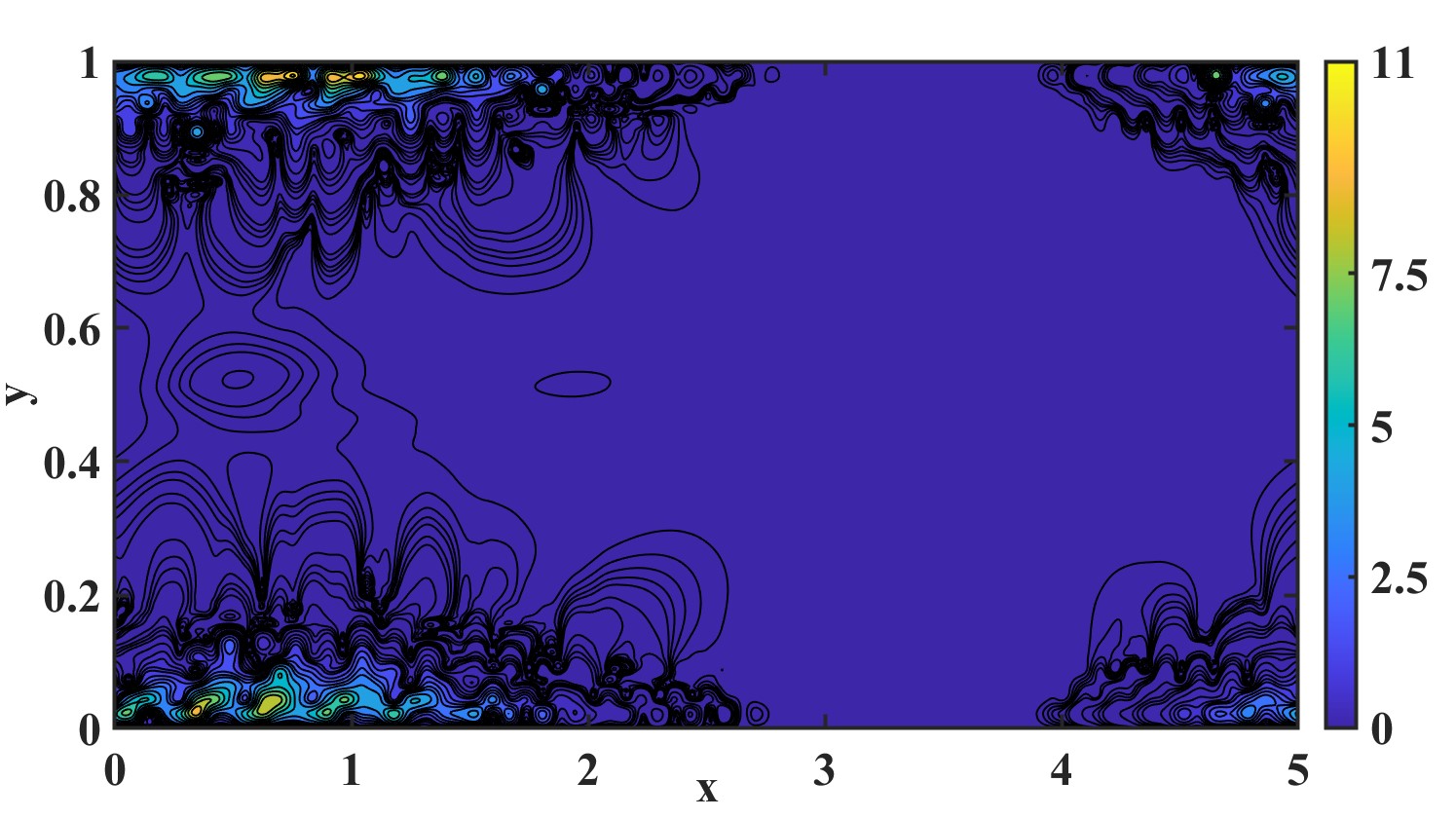}}\\
\subfloat[]{\label{fig2g}\includegraphics[width=0.49\linewidth, height=0.3\linewidth]{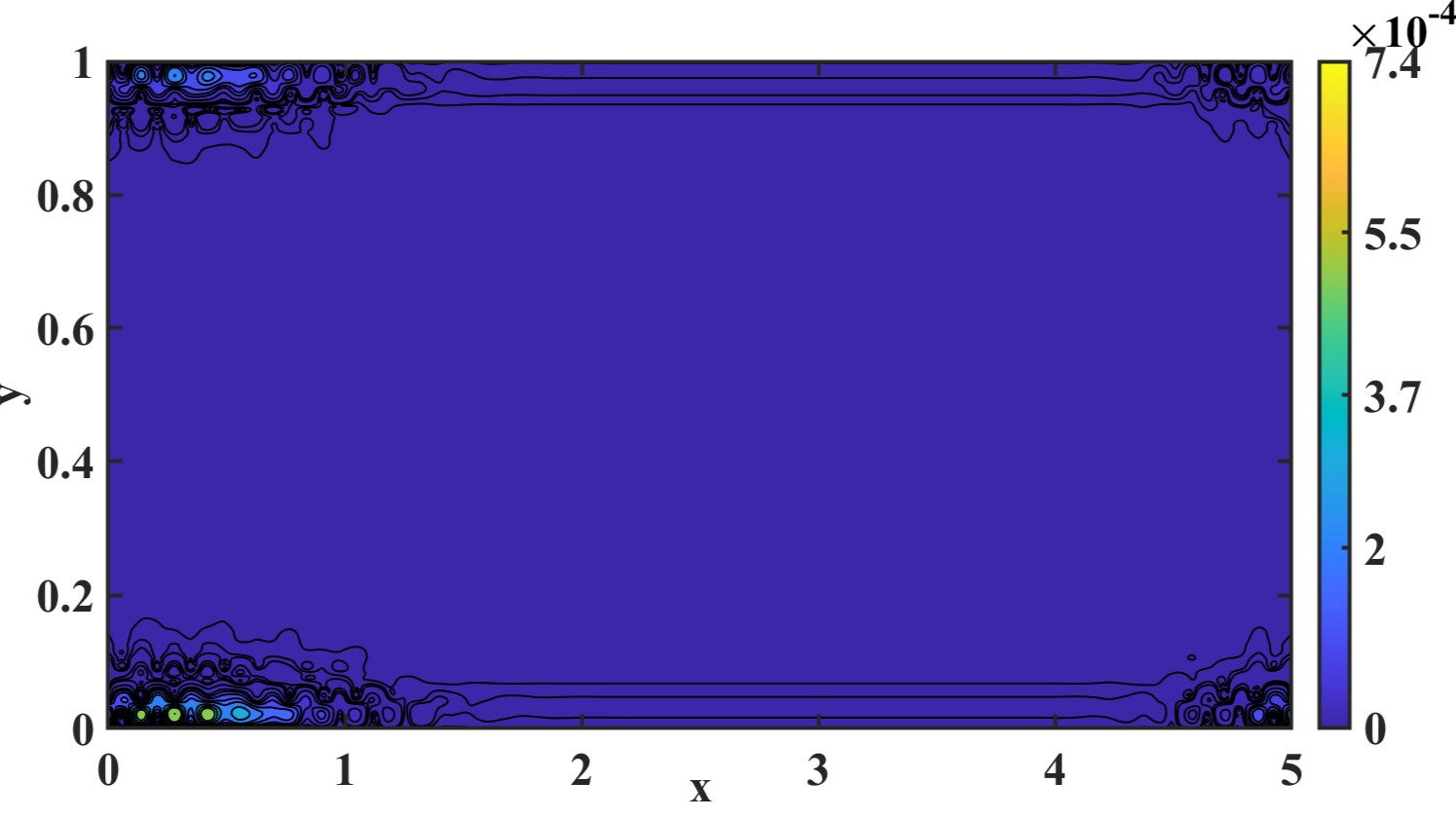}}
\hspace{0.1cm}
\subfloat[]{\label{fig2h}\includegraphics[width=0.49\linewidth, height=0.3\linewidth]{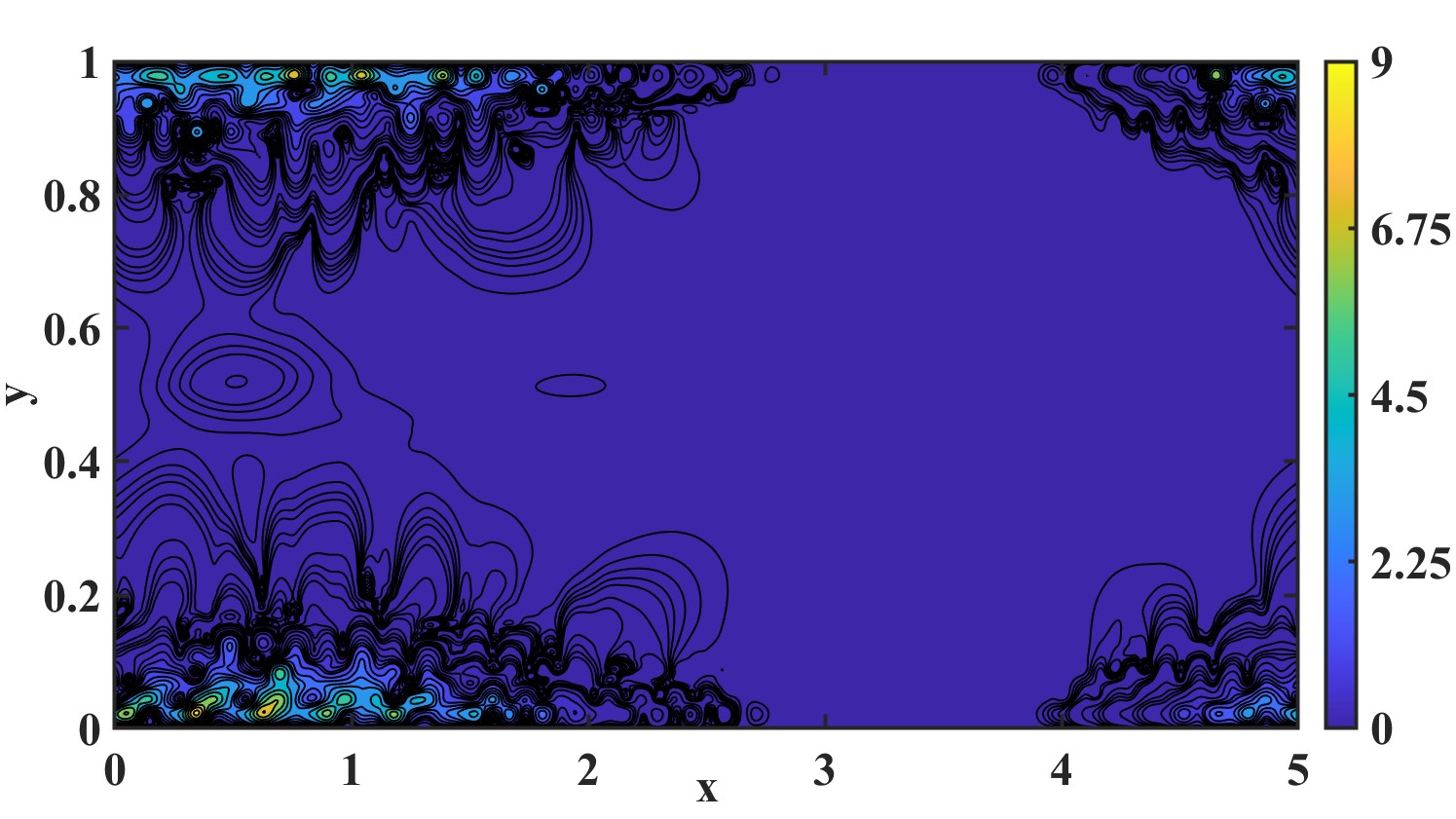}}
\caption{Contours of instantaneous (a, b) principle invariant of $\mathcal{G}$, $\tr \mathcal{G}$, (c, d) volume ratio, $\delta_1$, (e, f) shortest distance from mean, $\delta_2$, (g, h) anisotropy index, $\delta_3$, for the Zimm's model (left column), and the Rouse model (right column) at simulation time, $T=4.4$.}
\label{fig2}
\end{figure*}

\subsection{{\color{black} Coarse-grained} Rouse model}\label{subsec:Rouse}
The Rouse model~\cite{Rouse1953} predicts that the viscoelastic properties of the polymer chain via a generalized Maxwell model, where the elasticity is governed by a single relaxation time, which is independent of the number of Maxwell elements (or the so-called `submolecules'). The Rouse model represents `thicker' fluid, or fluids with slower diffusion than the Zimm's solution, due to the smaller fractional time-derivative ($\alpha = \case{1}{2}$, ~\fref{fig2} (right column)). Flows with smaller time-derivative (or the thicker polymer melt case), are those associated with higher concentration of polymers per unit volume. Experiments~\cite{Fogelson2015} have shown that non-Newtonian fluids with a larger polymer concentration, have a greater tendency (for the polymer strands) to agglomerate, the so-called `over-crowding effect'~\cite{Doi1996}. After comparing the respective range of all the invariants in both models, we find that our numerical simulations corroborate the experiments, namely: (i) the macrostructures in the Rouse model are more prominent, both in size as well as in magnitude (comparing ~\fref{fig2c} versus ~\fref{fig2d}), and (ii) the alternating regions of expansion interlaced with compression are more heterogeneous in the Rouse model (comparing ~\fref{fig2e},~\fref{fig2g} versus ~\fref{fig2f},~\fref{fig2h}, respectively). These observations indicate that the Rouse model is comparatively more unstable than the Zimm’s model, at the chosen values of the flow-material parameters.

\section{Concluding Remarks} \label{sec:conclusions}
In this paper, we have developed a mathematically consistent decomposition of the conformation tensor, $\mathcal{C}$, into the structure tensor, $\mathcal{G}$ (definition~\ref{def:ST}), for viscoelastic sub-diffusive flows, that resolves the difficulties associated with the traditional arithmetic decomposition. We characterized the fluctuations in $\mathcal{G}$ by using a geometry specifically constructed for $\mathbf{PS}_3$ and obtained three scalar measures: the volume ratio, $\delta_1$ (definition~\ref{def:delta1}), the shortest distance from the mean conformation, $\delta_2$ (definition~\ref{def:delta2}) and the anisotropy index, $\delta_3$ (definition~\ref{def:delta3}). The linear perturbation studies and the fully nonlinear simulations provided interesting insights about the instantaneous polymer conformation tensor that are not readily available from an arithmetic decomposition of $\mathcal{C}$, including: (i) evaluation of a (perturbation amplitude dependent) maximum time during which the linear perturbative solution can be well approximated by the weakly nonlinear solution, along the Euclidean manifold, (ii) a better resolution of the instantaneous regions of elastic shocks (which are alternating regions of expanded and compressed polymer volume, as compared with the volume of the mean conformation), (iii) a better measure to detect neighborhoods where the mean conformation tensor tends to be significantly different in comparison to the instantaneous conformation tensor, and (iv) a better representation of the proximity of the shape of the instantaneous conformation tensor, in comparison to the shape of the mean conformation tensor.

While the analysis presented here has delivered a general framework to provide a quantitative explanation of the previously published experimental findings~\cite{Riley1988,Nandagopalan2018,Zarabadi2018,Zarabadi2019}, the detailed physics of the flow induced structure formation in viscoelastic sub-diffusive flows is currently underway.

\ack 
T.C. and S.S. acknowledge the financial support of the Grant CSIR-SRF 09/1117(0012)/2020-EMR-I and DST CRG/2022/000073, respectively.

\appendix
\section{Linearized system of equations governing initial conditions for Equation~\eref{eqn:Pert6}} \label{appA}

Assuming a normal mode expansion for the perturbed field, $\phi_1 = \widetilde{{\color{black}\phi}(y)}\rme^{\rmi k x}$ (where $\mathbf{\phi}_1 = (\Omega_1, \psi_1, \mathcal{L}_{\mathcal{G}_1})$), equation~\eref{eqn:Pert6} reduces to
\begin{eqnarray}
\fl &Re\left( (y-y^2) \widetilde{\Omega} (\rmi k ) - 2 \widetilde{\psi} (\rmi k)\right) = \nu \left( \widetilde{\Omega} (\rmi k)^2 + \widetilde{\Omega}'' \right) + \frac{1-\nu}{We} \left(-k^2 F_1 F_2 \widetilde{\mathcal{L}_{\mathcal{G}_{11}}} - k^2  F_2^2 
\right. \nonumber\\
\fl & \left.
\widetilde{\mathcal{L}_{\mathcal{G}_{12}}} - k^2 F_1 F_4 \widetilde{\mathcal{L}_{\mathcal{G}_{12}}} - k^2 F_2 F_4 \widetilde{\mathcal{L}_{\mathcal{G}_{22}}} - {F_1}'' F_2 \widetilde{\mathcal{L}_{\mathcal{G}_{11}}} - F_1 {F_2}'' \widetilde{\mathcal{L}_{\mathcal{G}_{11}}} - F_1 F_2 \widetilde{\mathcal{L}_{\mathcal{G}_{11}}}'' - {F_2^2}'' 
\right. \nonumber\\
\fl & \left.
\widetilde{\mathcal{L}_{\mathcal{G}_{12}}} - F_2^2 \widetilde{\mathcal{L}_{\mathcal{G}_{12}}}'' \!-\!{F_1}'' F_4 \widetilde{\mathcal{L}_{\mathcal{G}_{12}}} - F_1 {F_4}'' \widetilde{\mathcal{L}_{\mathcal{G}_{12}}} \!-\!F_1 F_4 \widetilde{\mathcal{L}_{\mathcal{G}_{12}}}'' - {F_2}'' F_4 \widetilde{\mathcal{L}_{\mathcal{G}_{22}}} \!-\!  F_2 {F_4}'' \widetilde{\mathcal{L}_{\mathcal{G}_{22}}} 
\right. \nonumber\\
\fl & \left.
\!-\! F_2 F_4 \widetilde{\mathcal{L}_{\mathcal{G}_{22}}}'' - 2{F_1}' {F_2}' \widetilde{\mathcal{L}_{\mathcal{G}_{11}}} - 2 {F_1}' F_2 \widetilde{\mathcal{L}_{\mathcal{G}_{11}}}' - 2 {F_1}' {F_4}' \widetilde{\mathcal{L}_{\mathcal{G}_{12}}} - 2 {F_1}' F_4 \widetilde{\mathcal{L}_{\mathcal{G}_{12}}}' - 2 F_1 {F_2}' 
\right. \nonumber\\
\fl & \left.
\widetilde{\mathcal{L}_{\mathcal{G}_{11}}}' - 2 F_1 {F_4}' \widetilde{\mathcal{L}_{\mathcal{G}_{12}}}'  - 2 {F_2}' {F_4}' \widetilde{\mathcal{L}_{\mathcal{G}_{22}}} - 2 {F_2}' {F_4} \widetilde{\mathcal{L}_{\mathcal{G}_{22}}}' - 4 F_2 {F_2}' \widetilde{\mathcal{L}_{\mathcal{G}_{12}}}'- 2 F_2 {F_4}' \widetilde{\mathcal{L}_{\mathcal{G}_{22}}}' + 
\right. \nonumber\\
\fl & \left.
2 \rmi k F_2 {F_2}' \widetilde{\mathcal{L}_{\mathcal{G}_{11}}} + \rmi k  F_2^2 \widetilde{\mathcal{L}_{\mathcal{G}_{11}}}' + 2 \rmi k {F_2}' F_4 \widetilde{\mathcal{L}_{\mathcal{G}_{12}}} + 2 \rmi k F_2 {F_4}' \widetilde{\mathcal{L}_{\mathcal{G}_{12}}} + 2 \rmi k F_2 F_4 \widetilde{\mathcal{L}_{\mathcal{G}_{12}}}' + 2 \rmi k 
\right. \nonumber\\
\fl & \left.
F_4 {F_4}' \widetilde{\mathcal{L}_{\mathcal{G}_{22}}} + \rmi k F_4^2 \widetilde{\mathcal{L}_{\mathcal{G}_{22}}}' - 2 \rmi k F_1 {F_1}' \widetilde{\mathcal{L}_{\mathcal{G}_{11}}} - \rmi k F_1^2 \widetilde{\mathcal{L}_{\mathcal{G}_{11}}}' - 2 \rmi k {F_1}' F_2 \widetilde{\mathcal{L}_{\mathcal{G}_{12}}} - 2 \rmi k  F_1 {F_2}'  
\right. \nonumber\\
\fl & \left.
\widetilde{\mathcal{L}_{\mathcal{G}_{12}}} - 2 \rmi k F_1 F_2 \widetilde{\mathcal{L}_{\mathcal{G}_{12}}} - 2 \rmi k  F_2 {F_2}' \widetilde{\mathcal{L}_{\mathcal{G}_{22}}} - \rmi k F_2^2 \widetilde{\mathcal{L}_{\mathcal{G}_{22}}}' \right),
\label{eqn:A1}
\end{eqnarray}
%
\begin{eqnarray}
 \fl  \widetilde{\psi} (\rmi k)^2 + \widetilde{\psi}'' = -\widetilde{\Omega},
  \label{eqn:A2}
\end{eqnarray}

\begin{eqnarray}
\fl & (y-y^2) (\rmi k) \widetilde{\mathcal{L}_{\mathcal{G}_{11}}} = \frac{2}{\sqrt{d}} \left(F_1 F_4 {\widetilde{\psi}}' (\rmi k) - F_1 F_2 {\widetilde{\psi}}'' - F_2 F_4 {\widetilde{\psi}} (\rmi k)^2 + F_2^2 {\widetilde{\psi}}' (\rmi k) -
\right. \nonumber\\
\fl & \left.
\widetilde{\mathcal{L}_{\mathcal{G}_{11}}} F_1 F_2 (1-2y) - F_2^2 \widetilde{\mathcal{L}_{\mathcal{G}_{12}}} (1-2y) - {\widetilde{\psi}}(\rmi k) \left(-F_4 F_1' + F_2 F_2'\right) \right)- \frac{\widetilde{\mathcal{L}_{\mathcal{G}_{11}}}}{We},
    \label{eqn:A3}
\end{eqnarray}

\begin{eqnarray}
\fl & (y\!-\!y^2) (\rmi k) \widetilde{\mathcal{L}_{\mathcal{G}_{12}}} \!=\! \frac{1}{\sqrt{d}} \left(- 2 F_1 F_2 {\widetilde{\psi}}' (\rmi k) \!+\! F_1^2 {\widetilde{\psi}}'' \!+\! F_2^2 {\widetilde{\psi}} (\rmi k)^2  \!+\! \widetilde{\mathcal{L}_{\mathcal{G}_{11}}} F_1^2 (1\!-\!2y) \!+\!
\right. \nonumber\\
\fl & \left.
2 F_2 F_4 {\widetilde{\psi}}'(\rmi k) - F_2^2 {\widetilde{\psi}}'' - F_4^2 {\widetilde{\psi}} (\rmi k)^2 - \widetilde{\mathcal{L}_{\mathcal{G}_{22}}} F_2^2 (1-2y) - {\widetilde{\psi}} (\rmi k ) \left(F_2 F_1' - F_1 F_2' - F_4
\right.\right. \nonumber\\
\fl & \left. \left.
F_2' + F_2 F_4'\right)\right) - \frac{\widetilde{\mathcal{L}_{\mathcal{G}_{12}}}}{We},
\label{eqn:A4}
\end{eqnarray}

\begin{eqnarray}
\fl & (y\!-\!y^2)(\rmi k) \widetilde{\mathcal{L}_{\mathcal{G}_{22}}} \!=\! \frac{2}{\sqrt{d}} \left( \widetilde{\mathcal{L}_{\mathcal{G}_{12}}} F_1^2 (1\!-\!2y) \!-\! F_2^2 {\widetilde{\psi}}' (\rmi k)  \!+\! F_1 F_2 {\widetilde{\psi}}'' \!+\! F_2 F_4 {\widetilde{\psi}} (\rmi k)^2 - 
\right. \nonumber\\
\fl & \left.
F_1 F_4 {\widetilde{\psi}}' (\rmi k)   + \widetilde{\mathcal{L}_{\mathcal{G}_{22}}} F_1 F_2 (1-2y) - {\widetilde{\psi}} (\rmi k) \left(F_2 F_2' - F_1 F_4'\right)\right) - \frac{\widetilde{\mathcal{L}_{\mathcal{G}_{22}}}}{We},
 \label{eqn:A5}
\end{eqnarray}

%
%
%
%
%
%
%
where we denote $\frac{d}{d y}(\,) = (\,\,)'$ and
\begin{eqnarray}
   \overline{\mathbf{F}} =  
   \left[ \matrix{ F_1 & F_2 \cr F_2 & F_4 \cr}\right]
 =
\left[\matrix{
   \frac{1+\sqrt{d}}{\sqrt{2d + 2 \sqrt{d}}} &
   \frac{We (1-2y)}{\sqrt{2d + 2 \sqrt{d}}} \cr
   \frac{We (1-2y)}{\sqrt{2d + 2 \sqrt{d}}} &
   \frac{2d + \sqrt{d}-1}{\sqrt{2d + 2 \sqrt{d}}} \cr}\right]
\end{eqnarray}
where
\[
d = 1 + {We}^2{(1-2y)}^2.
\]

The solution to the boundary value problem is found subject to the boundary conditions, $(\widetilde{\psi}(y),\widetilde{\psi'}(y)) = (0, 0)$ at the rigid walls $y = 0, 1$.

\section*{References}
\bibliographystyle{bibstyle_iop}
\bibliography{references.bib} 

\end{document}